%% file: MainV2.tex
\documentclass[journal,onecolumn]{IEEEtran}
\usepackage[font=footnotesize]{caption}
\usepackage{graphicx,xcolor}
\usepackage{amsfonts,amsmath,amssymb}
\usepackage{amsthm}
\usepackage{amsmath}
\usepackage{comment}
\usepackage{dsfont}
\usepackage{cite}
\usepackage{subcaption}
\usepackage{subfig}
\ifCLASSINFOpdf
 \else
\fi
\newtheorem{theorem}{Theorem}
\begin{document}
\allowdisplaybreaks
\include{define.tex}
\title{Target Detection in Clustered Mobile Nanomachine Networks}

\author{Nithin V. Sabu, Kaushlendra Pandey, Abhishek K. Gupta, Sameer S.M.
\thanks{ N. V. Sabu and Sameer S.M. are with the National Institute of Technology Calicut, Kozhikode, Kerala 673601, India (Email:{ \{nithinvs,sameer\}@nitc.ac.in}).}
\thanks{K. Pandey is with Central Institute of Technology Kokrajhar, Assam 783370 India (Email:kk.pandey@cit.ac.in).}
\thanks{ A. K. Gupta is with Indian Institute of Technology Kanpur, Kanpur UP 208016, India (Email:gkrabhi@iitk.ac.in).}}

\markboth{}%
{Shell \MakeLowercase{\textit{et al.}}: Bare Demo of IEEEtran.cls for IEEE Journals}

\maketitle

\begin{abstract}
This work focuses on the development of an analytical framework to study a diffusion-assisted molecular communication-based network of nano-machines (NMs) with a clustered initial deployment to detect a target in a three-dimensional (3D) medium. Leveraging the Poisson cluster
process to model the initial locations of clustered NMs, we derive the analytical expression for the target detection probability with respect to time along with relevant bounds. We also investigate a single-cluster scenario. All the derived expressions are validated through extensive particle-based simulations. Furthermore, we analyze the impact of key parameters, such as the mean number of NMs per cluster, the density of the cluster, and the spatial spread, on the detection performance. Our results show that detection probability is greatly influenced by clustering, and different spatial arrangements produce varying performances. The results offer a better understanding of how molecular communication systems should be designed for optimal target detection in nanoscale and biological environments.
\end{abstract}


\IEEEpeerreviewmaketitle
\section{Introduction}\label{Sec:intro}

Molecular communication (MC) is a promising approach to communication that draws inspiration from biological systems, enabling information transfer via molecules \cite{nakano2013c, farsad2016}. In addition to various applications of MC including targeted drug delivery \cite{chude-okonkwo2017}, environmental monitoring \cite{nakano2012b}, and biosensing, target detection is an important application that deals with the identification of targets of interest, such as harmful pathogens or biomarkers associated with diseases. Target detection techniques using nanomachines (NMs) can be particularly useful in medical scenarios such as early detection of cancerous cells or infectious agents. Due to the limited functionality of a single NM, deploying multiple NMs in the environment is essential to ensure adequate detection coverage. This necessity highlights the importance of mathematical modeling and system-level analysis for large-scale target detection systems comprising multiple NMs.

The past literature has attempted to model and analyze target detection in MC systems under various MC settings. For example, in \cite{nakano2017a}, the authors studied a leader-follower model for mobile bio-nanomachines designed to detect targets. A target detection system where stationary NMs detect continuously emitting targets in the surrounding medium was studied in \cite{mosayebi2018}.  Target detection and tracking in a two-dimensional context using stationary and mobile NMs was investigated in \cite{okaie2016a}. 

A large-scale system of fully absorbing NMs can be modeled as spheres with centers distributed spatially according to a Poisson point process (PPP) \cite{andrews2023a, haenggi2012}. Homogeneous PPPs have been widely used in molecular communication research to capture the inherent randomness of NM spatial distributions \cite{sabu2019, deng2017a, dissanayake2019a}. In \cite{sabu2020b}, the detection probability for a mobile target molecule interacting with multiple stationary NMs distributed as PPP was derived.  Based on \cite{sabu2020b}, the work in \cite{sabu2024d} provides an analytical framework for modeling and analyzing the performance of target detection systems that involve multiple mobile NMs of varying sizes with passive or absorbing boundaries. The framework evaluates detection performance for various scenarios, including degradable and non-degradable targets as well as mobile and stationary targets. These works considered a uniform spatial distribution of NMs.

Clustering is a common phenomenon in biological systems, where spatially localized interactions are essential to maintain functionality. Clustering can occur naturally due to spatial limitations, resource availability, or specific deployment strategies. NMs may initially form clusters in biological settings depending on controlled release mechanisms, external forces, or natural aggregation.
 In targeted drug delivery, for instance, nanoparticles containing therapeutic agents are sometimes delivered close to sick tissues, where they form clusters before diffusing toward their target.  In biosensing applications, artificial nanosensors or engineered bacteria could be used in dense colonies at designated locations.  These clusters then disperse via Brownian motion to detect biological or chemical markers.

The Poisson cluster process (PCP) is widely used to model clustering in wireless networks due to its mathematical tractability, mainly attributed to its known probability generating functional (PGFL) \cite{haenggi2012}. 
 In a PCP-based model, the parent points representing the cluster centers are distributed according to PPP, while each cluster contains nodes distributed as a daughter point process (PP) around the parent points. Depending on the choice of the daughter PP, PCP can be classified into Mat\'ern cluster process (MCP) or Thomas cluster process (TCP), making it a flexible model to capture the spatial characteristics of clustered NM networks \cite{haenggi2012}.
Leveraging the mathematical tractability of PCP models, several studies have derived key spatial properties of PCP, such as the contact distance and nearest-neighbor distance distributions. Specifically, the MCP for these metrics has been analyzed in \cite{Azimi2018SingleCluster,Afshang2017Matern,Pandey2021Matern}, while similar results have been explored for the TCP in \cite{Afshang2017Thomas}.  Given its tractability and well-studied properties, PCP provides a suitable framework for modeling clustered networks of NMs.

Recently, stochastic geometry-based models were studied for clustered MC systems \cite{azimi-abarghouyi2023,azimi2024matern}. For example, \cite{azimi-abarghouyi2023} proposed a framework for modeling and analyzing clustered molecular nanonetworks. They derived distance distributions in a three-dimensional (3D) space, leading to simplified expressions for the TCP, and examined interference and error probabilities. Additionally, \cite{azimi2024matern} introduced a novel variant of the MCP, termed the MCP with holes (MCP-H).  The authors characterized conditional distance distributions and derived simple closed-form expressions, enabling further analysis of the PGFL. Another notable work is \cite{albeverio1998}, in which the survival probability of a Brownian particle was analyzed in a Poisson cluster process of correlated and stationary traps.

{\color{black}While prior work, such as \cite{sabu2024d}, has analyzed target detection using a PPP for uniform NM deployment, research on clustered target detection systems is lacking. Clustering is prevalent in practical applications like targeted drug delivery, where NMs may be released in localized groups, and extending PPP to clustered models like PCP is non-trivial due to the need to handle spatial correlations. Our work fills this gap by modeling and analyzing target detection in systems where spherical mobile NMs are spatially clustered using a PCP. We derive exact analytical expressions for detection probability in clustered settings, along with bounds and approximations for clustered NMs, and investigate single-cluster scenarios to quantify the impact of clustering parameters (e.g., cluster density $\lambda_\Pa$, mean NMs per cluster $\mbar$, spatial spread $r$ or $\sigma$) on detection performance. The single-cluster analysis is also distinct, as it models a non-stationary distribution where NMs are localized around a single cluster center, unlike the stationary PPP or PCP models. The key contributions of this work are summarized as follows:}

\begin{itemize}
\item We develop an analytical framework for evaluating the
target detection performance of a diffusion-assisted molecular communication-based network with NMs with clustered initial deployment.
    \item We derive analytical expressions for the target detection probability of NMs spatially distributed using PCP, specifically considering the MCP and TCP models.
    \item We also characterize the mean number of clusters that detect the target, which can provide us an alternative expression for target detection probability. Using it, we present the upper and lower bounds for the detection probability of PCP-deployed NMs.
    \item We derive analytical expressions for the detection probability at $t=0$ for MCP and TCP, capturing the impact of initial spatial distribution.
    \item We analyze a special case where NMs are deployed in a single cluster and derive detection probabilities for both uniform and Gaussian-distributed NMs.
    \item All analytical derivations are validated using particle-based simulations, demonstrating their accuracy under various parameter settings.
    \item We compare the detection probability of NMs deployed using PCP (MCP and TCP) against the PPP and a single-cluster deployment, providing insight into the trade-offs between clustering and uniform spatial distributions.
\end{itemize}
This work enhances the understanding of how clustering affects molecular communication by examining the effects of various system parameters, thus providing useful design insights for nanoscale and biological applications.

\textbf{Notation:} 
The Euclidean norm of a vector is denoted by  $\|.\|$. The spherical ball of radius $r$ centered at $\x$ is represented as $\ball{\x,r}$. The volume of a region $R$ is denoted by $|R|$. $\mathsf{A}\oplus\mathsf{B}$ represents the Minkowski sum of sets $\mathsf{A}$ and $\mathsf{B}$.

\begin{figure*}[ht]
    \centering
    \includegraphics[width=0.8\linewidth]{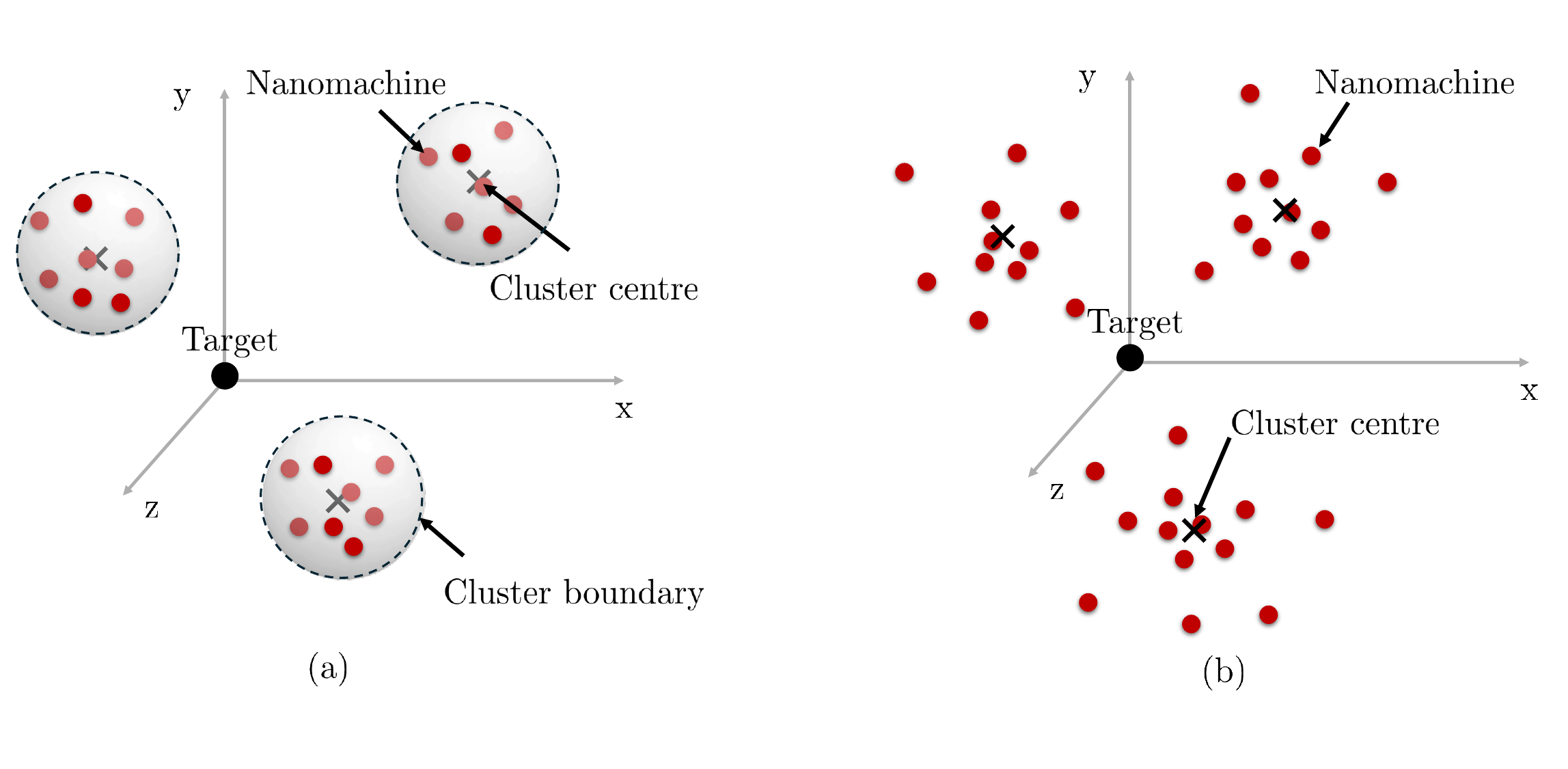}
    \caption{Schematic diagrams of (a) MCP and (b) TCP. Crosses ($\times$) represent cluster centers (parent points), and red dots (${\rm o}$) represent daughter points (NMs). }
    \label{MCP_TCP_Fig}
\end{figure*}
\section{System Model}\label{Sec:SysModel}
{\color{black}In this work, we consider a molecular communication-based target detection system where multiple spherical NMs of radius $a$ are deployed around static cluster centers in a 3D medium to detect a point target.} To model the clustered NMs, in this work, we can utilize the PCP, which we describe first in the following subsection.

\subsection{Poisson Cluster Process (PCP)}\label{SubSec:PCP}
PCP based models can capture clustered patterns in various applications, such as spatial patterns in biology \cite{azimi-abarghouyi2023}, clustered wireless communications \cite{saha2017enriched}, and environmental modeling. A PCP consists of points/nodes centered around cluster centers. In a PCP, the locations of cluster centers are distributed according to a PPP $\Phi_\Pa = \{\x_1, \x_2, \dots\}$ with density $\lambda_\Pa$. Each cluster consists of points/nodes that are identically distributed around the respective cluster center $\x_i \in \Phi_\Pa$. Let $\mathcal{N}_i$ denote the points in the $i$th cluster around cluster center $\x_i$. The point processes formed by points in each cluster can be seen as a daughter point process with the cluster center as its parent point. Further, let $\y_{ij}\in \mathcal{N}_i$ denote the location of the $j$th daughter point in the $i$th cluster. Hence, the PCP $\Phi_\Pcp$ can be expressed as,
\begin{align} 
	\Phi_{\Pcp} = \bigcup\nolimits_{i \in \mathbb{N}} \x_i + \mathcal{N}_i.
\end{align}
If the daughter points are also distributed as an independent PPP, then the cluster process is called \textit{doubly PCP} \cite{haenggi2012}. Here, the number of points in each cluster is Poisson distributed with mean $\mbar$. It is widely used in modeling of various phenomena in which events or objects occur in localized groups around central points.
{\color{black}Note that the cluster density $\lambda_\Pa$ denotes the mean number of cluster centers per unit volume, while the NM density $\lambda_{\diff}$ represents the mean number of NMs per unit volume. Both of the densities are related by $\lambda_{\diff} =\mbar \lambda_\Pa$. Thus, cluster density describes the spatial distribution of clusters, whereas NM density quantifies the overall node population in the network.}

  Two important variants of doubly PCP are MCP and TCP, which are frequently used in applications that require spatial modeling of clustered events. 

\subsubsection{MCP}
MCP is a doubly PCP with daughter points uniformly distributed within a spherical region of radius $r$ around their parent point. The MCP can be defined as,
\begin{align}
	\Phi_\Mat=\bigcup\nolimits_{i\in \mathbb{N}} \x_i+\mathcal{M}_i,
\end{align}
where  $\mathcal{M}_i$ denotes the $i$th cluster (relative to the cluster center location) with its points $\y_{ij}\in\mathcal{M}_i$  uniformly distributed in a spherical ball $\ball{\origin,r}$ with the probability density function (PDF) given by
\begin{align}
	f(\y)&=\frac{1}{(4/3) \pi r^3}\mathds{1}\left(\y\in \ball{\origin,r}\right),\label{Matdens}
\end{align}  independent of other points. Hence, each daughter PP $\mathcal{M}_i$ is a PPP with density  $\lambda(\y)=\mbar f(\y)$ and the density of the whole clustered PP is $\lambda_\Mat=\lambda_\Pa \mbar$.

\subsubsection{TCP}
TCP is another variant of doubly PCP where the distance of each daughter point from its parent is distributed as a normal distribution with zero mean and variance $\sigma^2$. TCP is suitable for modeling scenarios where the daughter points are more spread with $\sigma^2 $ determining the spread of a cluster. The TCP can be defined as,
\begin{align}
	\Phi_\Tho=\bigcup\nolimits_{i\in \mathbb{N}} \x_i+\mathcal{T}_i,
\end{align}where $\mathcal{T}_{i}$ denotes the $i$th cluster (relative to the cluster center location) with its points $\y_{ij}$, $j \in \mathbb{N}$ normally distributed with PDF,
\begin{align}
	f(\y) =\frac{1}{(2\pi\sigma^2)^{3/2}}\exp\left(-\frac{\|\y\|^2}{2\sigma^2}\right) \mathds{1}\left(\|\y\|\geq 0\right)\label{Thodens}.
\end{align}Here, $\|\y\|$ represents the Euclidean norm. Similar to MCP, each daughter PP $\mathcal{T}_i$ is a PPP with density  $\lambda(\y)=\mbar f(\y)$ and  the overall density of the TCP is $\lambda_\Tho=\lambda_\Pa \mbar$. Fig. \ref{MCP_TCP_Fig} (a) and (b) show a realization of MCP and TCP, respectively.

\subsection{Target Detection Framework}

In this work, we consider a target detection system where NMs are deployed in clustered fashion to detect a specific target. The NMs are mobile with their motion governed by Brownian dynamics. Below, we detail the key components of this system, including the nature of cluster centers, the target, and the detection mechanism.

{\color{black} We consider that the NMs are initially deployed as a PCP, forming clusters around designated cluster centers. Note that cluster centers described here do not represent physical devices. They act as virtual reference points for clusters and help define the spatial distribution of NMs. Note that for each cluster, its NMs are distributed around the corresponding cluster center. The primary goal of the system is to detect a target that can be thought of as a biological or chemical entity of interest, like a toxic molecule, pathogen, or disease biomarker. A single target is assumed for analytical tractability, with extensions to multiple targets deferred to future work. } 
 After deployment, the NMs start to propagate in the medium. To represent the region covered by NMs, we associate a spherical region of radius $a$ with each daughter point representing the center of the corresponding NM. Hence, the overall region covered by NMs at time $t = 0$ can be written as the union of all such spheres given as
\begin{align}\label{onlypsi}
\Psi &= \bigcup\nolimits_{\x_i \in \Phi_\Pa} \bigcup\nolimits_{\y_{ij} \in \mathcal{N}_i} \left( \x_i + \y_{ij} + \ball{\origin,a} \right)\nonumber\\
&= \bigcup\nolimits_{\x_i \in \Phi_\Pa} \bigcup\nolimits_{\y_{ij} \in \mathcal{N}_i}  \ball{\x_i + \y_{ij} ,a},
\end{align}
which we term as {\em Boolean PCP}.

After initial deployment $(t=0)$, the NMs propagate through the medium via Brownian motion with diffusion coefficient $D$. The event of interest is when any of the NMs detect the target by touching or intersecting it. Without loss of generality, let us take the target at the origin.
 Let $\bps_{ij}(t)=\{\z_{ij}(s),\ 0\leq s\leq t\}$ represent the relative Brownian path of a single NM at $\x_i+\y_{ij}$ over time $t$ starting from the origin $\origin$. Here, $\z_{ij}(s)$ is the location of the center of that NM at time $s$. Now, we can define $\bp_{ij}(t)$ as the path of the NM at $\x_i+\y_{ij}$ until time $t$, \ie $\bp_{ij}(t)=\x_i+\y_{ij}+\bps_{ij}(t)$. Similar to above,  the region covered by NMs within time $t$ can be given as
 \begin{align}
     \Psi_t= \bigcup\nolimits_{\x_i \in \Phi_\Pa} \bigcup\nolimits_{\y_{ij} \in \mathcal{N}_i} \bp_{ij}(t)\oplus \ball{\origin,a}.\label{eq:psit}
 \end{align}

The event $\Event$  of  hitting of a point target at origin $\origin$ by any of the NMs is given as
 \begin{align}
     \Event &=\{\Psi_t\cap \{\origin\} \neq \phi)\} \nonumber\\
     &=\bigcup\nolimits_{\x_i \in \Phi_\Pa} \bigcup\nolimits_{\y_{ij} \in \mathcal{N}_i} \Event_{ij},\label{event1}
 \end{align}
 where $\Event_{ij}=\{ \bp_{ij}(t)\oplus \ball{\origin, a} \cap \{\origin\} \neq \emptyset \}$. 
 Note that $\bp_{ij}(t) \oplus \ball{\origin, a}$ in \eqref{event1} represents the region covered up to time $t$ by the $j$-th NM of the $i$-th cluster. $\Event_{ij}$ denotes the event that this region intersects with the target point at the origin $\origin$, which means that the target is considered detected by the ($i,j)$th NM.
 
{\color{black} To compute the performance of the system, we define and utilize the metric detection probability. The detection probability is the probability that at least one NM, deployed in the medium, intersects with a point target within time $t$.}
 
 \section{Target Detection Via A Network of Clustered NMs}\label{Sec:TDCluster}
In this section, we study the target detection capability of a network of NMs, initially deployed as a PCP. We first derive the detection probability as given in the following theorem.

 \begin{theorem} \label{Thm_actual}
  The probability of detecting a point target at the origin by any of the NMs within time $t$, with the initial locations of NMs distributed as a PCP is (for proof, see Appendix \ref{Ap_actual})
\begin{align}
    \Tdp_{\Pcp}(\lambda_\Pa,t)&=
1-\exp\left(-\lambda_\Pa \int_{\rthree}\left(1-\exp\left(-\mbar\int_{\rthree}\frac{a}{||\x+\y||}\right.\right.\right. \left.\left.\erfc\left(\frac{||\x+\y||-a}{\sqrt{4Dt}}\right)f(\y)\diff \y\bigg)\right)\diff \x\right). \label{TDact}
\end{align} 
 \end{theorem}The target detection probability for special cases where NMs are initially deployed as MCP or TCP, are given in the following Corollaries 1 and 2, respectively. 
\begin{corollary}\label{MCPcoro}
Under initial MCP based deployment of NMs, the target detection probability within time $t$ is given by
\begin{align}
    &\Tdp_\Mat (\lambda_\Pa,t)=
1-\exp\left(-4\pi\lambda_\Pa \int_{0}^\infty\left(1-\exp\left(-\frac{3\mbar}{2 r^3}\int_{0}^r\int_0^\pi\right.\right.\right. \left.\left.\frac{a}{\gamma(x,y,\theta)}\erfc\left(\frac{\gamma(x,y,\theta)-a}{\sqrt{4Dt}}\right)y^2\sin(\theta)\diff \theta\diff y\bigg)\right)x^2\diff x\right),\label{MCPtarget}\\
&\text{where, }\gamma(x,y,\theta)=\sqrt{x^2+y^2+2xy\cos(\theta)}.\nonumber
\end{align} 
\begin{proof}
By substituting the value of $f(\y)$ from \eqref{Matdens} in \eqref{TDact} we get,
    \begin{align}
    &\Tdp_\Mat (\lambda_\Pa,t)=
1-\exp\left(-\lambda_\Pa \int_{\rthree}\left(1-\exp\left(-\frac{\mbar}{\frac43 \pi r^3}\int_{\ball{\origin,r}}\right.\right.\right. \left.\left.\frac{a}{||\x+\y||}\erfc\left(\frac{||\x+\y||-a}{\sqrt{4Dt}}\right)\diff \y\bigg)\right)\diff \x\right).\label{MCPtargettemp}
\end{align} 

{\color{black}
		We now use the isotropy of the parent point process to simplify the detection probability expression. Without loss of generality, we fix the location of a parent NM cluster center along the positive \(z\)-axis as \(\x = (0,0,x)\), where \(x = \|\x\| \geq 0\). This is justified because the isotropy (rotational invariance) of the point process ensures that the value of the integral over \(\x\) depends only on its radius \(x\), not on angular coordinates.
		
		Expressing the outer integral over \(\mathbf{x} \in \mathbb{R}^3\) in spherical coordinates \(\mathbf{x} = (x, \theta_x, \phi_x)\) with
		\[
		x \geq 0, \quad \theta_x \in [0, \pi], \quad \phi_x \in [0, 2\pi),
		\]
		the rotational symmetry allows integrating out the angular terms \(\theta_x\) and \(\phi_x\), resulting in a factor of \(4\pi\).
		
		Similarly, for the inner integral over daughter points \(\mathbf{y}\) uniformly distributed within the ball \(\mathcal{B}(\mathbf{0}, r)\), we use spherical coordinates \(\mathbf{y} = (y, \theta, \phi)\) with
		\[
		y \in [0, r], \quad \theta \in [0, \pi], \quad \phi \in [0, 2\pi).
		\]
			
		The Euclidean norm of the vector sum \(\mathbf{x} + \mathbf{y}\) is thus
		\[
		\|\mathbf{x} + \mathbf{y}\| = \sqrt{x^2 + y^2 + 2 x y \cos\theta} \equiv \gamma(x,y,\theta).
		\]
		Since this term is independent of the  $\phi$ of $\y$, the integration over $\phi$ in the inner integral yields a factor of $2\pi$. Combining these simplifications, we derive equation~\eqref{MCPtarget}, thus completing the proof.}
\end{proof}
\end{corollary}
 
 \begin{corollary}\label{TCPcoro}
     The target detection probability within time $t$ when the location of NMs modeled as TCP is given by
     \begin{align}\label{TDPTCP}
\Tdp_\Tho (\lambda_\Pa,t)&=
1-\exp\left(-4\pi \lambda_\Pa \int_{0}^\infty\left(1-\exp\left(-\frac{\mbar}{\sqrt{2\pi}\sigma^3}\int_{0}^\infty\right.\right.\right.\left.\left.\int_0^\pi\frac{a}{\gamma(x,y,\theta)}\erfc\left(\frac{\gamma(x,y,\theta)-a}{\sqrt{4Dt}}\right)e^{-\frac{y^2}{2\sigma^2}}y^2\sin(\theta)\diff \theta\diff y\bigg)\right)\right.\nonumber\\
&\left.\qquad \times x^2\diff x\right).
\end{align} 
\begin{proof}
    The proof is similar to that of Corollary \ref{MCPcoro}, hence it is skipped for brevity.
\end{proof}

 \end{corollary}
 Note that Corollary \ref{MCPcoro} and \ref{TCPcoro} give exact expressions of detection probability for specific cases with MCP and TCP based deployments.
 \subsection{Number of Clusters Detecting the Target}
We now present an alternative approach to get the expression for the target detection probability  presented in Theorem \ref{Thm_actual}. We first present the mean number of clusters that detect the target at the origin within time $t$ as given in the following lemma. 
 \begin{lemma}\label{Ap_Thmean1}
For a system with NMs initially distributed
as a PCP, the number of clusters that detect the target at the origin is a Poisson random variable with mean given as (For complete proof, see Appendix \ref{Ap_Thmean}) 
\begin{align}
         N_\Psi=\lambda_\Pa V(t),
     \end{align}
      where 
      \begin{align}
      	V(t)&=\Expect_{\mathcal{N}_\origin,\bps_{j}}\left[\bigm|  \bigcup\nolimits_{\y_{j}\in\mathcal{N}_\origin} \y_{j} + \bps_{j}(t)  \oplus \ball{\origin, a}\bigm| \right]\label{eq:V(t)ns}\\
      	&=\int_{\rthree}\left(1-\exp\left(-\mbar\int_{\rthree}\frac{a}{||\x+\y||}\right.\right. \left.\times\erfc\left(\frac{||\x+\y||-a}{\sqrt{4Dt}}\right)f(\y)\diff \y\bigg)\right)\diff \x \label{eq:V(t)}
      \end{align}
       denotes the volume of the overall region covered by NMs of the typical cluster ($\mathcal{N}_\origin$) within time $t$.
 \end{lemma}
{ The proof of the above lemma is obtained by treating the region covered by the NMs as a Boolean model with cluster centers as germs and combined region covered by NMs of each cluster as its grain (inner union in \eqref{eq:psit}).
	
 The target is detected when $N_\Psi$ is not zero. The resulting detection probability is therefore given by the following theorem.
 \begin{theorem}\label{meanpro}
{\color{black}The probability of detecting a point target located at the origin within time $t$, for a network of NMs initially distributed as a PCP, can be expressed in terms of the expected volume covered by the NMs of a typical cluster as}
     \begin{align}
   & \Tdp_{\Pcp}(\lambda_\Pa,t)=
1-\exp\left(-\lambda_\Pa V(t)\right). \label{TDact2}
\end{align} 

 \begin{proof}
     The number of clusters detecting the target is Poisson distributed with parameter $N_\Psi$. A cluster is said to detect the target if any of its NMs detects the target. Therefore, the target detection probability can be obtained from the void probability $1-\exp(-N_\Psi)$.
 \end{proof}
 \end{theorem}
 {\color{black}While the expressions presented in Theorems \ref{Thm_actual} and \ref{meanpro} are the same, Theorem \ref{Thm_actual} uses the PGFL approach, which is a key tool in stochastic geometry. This approach provides a structured way to derive the detection probability by directly accounting for the random distribution of NMs, making it easier to extend the analysis to more complex scenarios, such as mobile target models. On the other hand, Theorem \ref{meanpro} uses a Boolean model perspective, which offers an intuitive understanding by viewing detection as the coverage of a target by the swept volume $V(t)$. This technique is useful to derive the bounds of the detection probability.}  Now we present simple closed-form bounds for the target detection probability.
  \subsection{Bounds on Target Detection Probability}
  The following lemma gives the bounds of the average volume covered by NMs in the typical cluster.
  \begin{lemma}\label{volineqlem}
 The average volume of the region covered by NMs in the typical cluster within time $t$ is bounded as (for proof, see Appendix \ref{VolineqAp}),
  \begin{align}
               &\left(1-\exp(-\mbar)\right)\times W(t)\leq 
V(t)\leq \mbar \times W(t),\label{volineqex} 
  \end{align}
where $W(t)$ denotes the average volume covered within time $t$ by
a spherical particle of radius $a$ under  Brownian motion.
  \end{lemma}
  From \cite{berezhkovskii1989, sabu2020b}, $W(t)$ is given as
  \begin{align}
  	W(t)&={(4\pi a^3)}/{3}+4\pi Da t+8a^2\sqrt{\pi D t}.\label{wsv}
  \end{align}
 Using Theorem \ref{meanpro} and Lemma \ref{volineqlem}, the upper and lower bound of the target detection probability can be obtained, which are given below.
\begin{corollary}
    \label{ThmPCPUB}
The upper bound of the detection probability of  a point target at the origin by any of the NMs within time $t$ for PCP is 
    \begin{align}
    \Tdp_{\Pcp}(\lambda_\Pa,t)&\leq
1-\exp\left(-\lambda_\Pa \mbar W(t)\right).\label{UBeq}
\end{align}
\end{corollary}

\begin{corollary}\label{ThmPCPLB}
 The lower bound of the detection probability for a point target at the origin, detected by any of the NMs within time $t$ under a PCP, is
\begin{align}
    \Tdp_{\Pcp}(\lambda_\Pa,t)&\geq
1 - \exp\left(-\lambda_\Pa  \left(1-\exp(-\mbar)\right) W(t)\right). \label{LBeq}
\end{align}

\end{corollary}

\subsection{Comparison of Target Detection Using Clustered NMs vs PPP-Distributed NMs}\label{PerfComp}
The spatial distribution of the NMs can significantly affect the efficiency of target detection. We now compare the detection performance of clustered NMs (with initial deployment as PCP) with that of non clustered NMs (with initial deployment as PPP). Let us consider a target detection system where NMs are initially deployed as homogeneous  PPP in 3D with density $\lambda_\Pa \mbar$. Its target detection probability can be computed  from Lemma 2 and Theorem 1 of \cite{sabu2024d} by substituting $\lambda_\Pa \mbar$ for NM density and setting $r=0$, and is given as
\begin{align}
    &\Dp_{\PPP}(\mbar\lambda_\Pa,t)=1-\exp\left(-\lambda_\Pa \mbar W(t)\right).\label{PPPmbar}
\end{align}
By comparing \eqref{UBeq} and \eqref{PPPmbar}, we can see that the detection probability for unclustered NMs as given in \eqref{PPPmbar} serves as an upper bound to the detection probability with clustered NMs with initial PCP deployment. This indicates that clustered deployment may degrade detection probability. {\color{black}The performance gap between PCP and PPP arises because PCP's clustered structure concentrates NMs around parent points, creating dense and sparse regions of NMs. This effect is more prominent for cases with low \(\lambda_\Pa\) or small \(r/\sigma\). This situation leads to detection gaps, reducing the likelihood of an NM being near the target compared to the uniform coverage of PPP with equivalent density \(\lambda_\Pa \mbar\). This performance gap is also evident in the numerical results (see Fig.~\ref{fig:comparison}).}

 Further, comparing \eqref{LBeq} and \eqref{PPPmbar}, the system with clustered NMs performs better than a PPP distributed detection system with density $\lambda_\Pa \left(1-\exp(-\mbar)\right) $, \ie a system with a single NM in each of the non-empty cluster centers. This tells us that having more NMs in clusters compared to a single NM improves the detection. However, if NMs are unclustered, then the detection probability further increases. 
Therefore, 
\begin{align}
    \Dp_{\PPP}( \left(1-\exp(-\mbar)\right)\lambda_\Pa ,t)\leq \Dp_{\Pcp}(\lambda_\Pa,t)\leq \Dp_{\PPP}(\mbar\lambda_\Pa,t).
\end{align}

\subsection{Approximate Target Detection Probability}\label{Subsec:Approx}
In the following theorem, we present an approximate expression for the target detection probability. {\color{black}The approximation is obtained by assuming that all the NMs of a cluster are initially (at time $t=0$) located at the respective cluster centers, \ie the PDF $f(\y)$ becomes a Dirac delta function $\delta(\y)$. For MCP and TCP, this translates to $r=0$ and $\sigma=0$, respectively.} This approximate detection probability remains the same for any doubly PCP, including MCP and TCP.
\begin{theorem}\label{ThApprox}
The approximate target detection probability for a system with NMs distributed as doubly PCP is 
\begin{align}
    \Tdp_{\Pcp}(\lambda_\Pa,t)&\approx
1-\exp\left(-4\pi \lambda_\Pa \int_{0}^\infty\left(1-\exp\left(-\mbar {a}/{x} \right.\right.\right. \left.\left.\erfc\left(\frac{x-a}{\sqrt{4Dt}}\right)\bigg)\right) x^2\diff x\right).\label{TDPapprox}
\end{align}
\end{theorem}

\begin{proof}
{\color{black}	When NMs are concentrated at their respective parent point (i.e., $r=0$ for MCP or $\sigma=0$ for TCP) at time $t=0$, the daughter points $\y_{ij}$ are effectively located at the origin relative to their parent point $\x_i$, so $\x_i + \y_{ij} \approx \x_i$, and thus $\|\x + \y\| \approx \|\x\|$. Substituting into \eqref{TDact}, the inner integral simplifies as
\begin{align}
	\Tdp_{\Pcp}(\lambda_\Pa,t)&\approx
	1-\exp\left(-\lambda_\Pa \int_{\rthree}\left(1-\exp\left(-\mbar\right.\right.\right. \left.\left.\frac{a}{||\x||}\erfc\left(\frac{||\x||-a}{\sqrt{4Dt}}\right)\bigg)\right)\diff \x\right).\label{TDPapprox2}
\end{align}
Changing to the polar coordinates gives \eqref{TDPapprox}}.
\end{proof}

{\color{black}Note that this approximation assumes NMs are initially co-located with their cluster centers, effectively treating each cluster as a single point with $\mbar$ NMs. This simplifies the double integral in \eqref{TDact} to a single integral, making numerical evaluation more tractable while remaining accurate for small $r$ or $\sigma$.}

\subsection{Target Detection Probability of NM Network Without Motion}
 Now, we focus on the probability of detection of the target under which the NMs remain stationary at their initial positions \ie they do not undergo Brownian motion. This scenario corresponds to analyzing the system at $t=0$ where each NM is considered static and occupies a spherical volume $\ball{\origin,a}$.

 In this scenario, the target detection event occurs if the target intersects with the region occupied by any NM at its initial location. 
 The event of interest, denoted as $\Event_\origin$, is defined as
\begin{align} \Event_{\origin} = 
 \bigcup\nolimits_{\x_i \in \Phi_\Pa} \bigcup\nolimits_{\y_{ij} \in \mathcal{N}_i} \bigg\{\left( \x_i + \y_{ij} + \ball{\origin,a} \right)\cap \origin
\neq \phi\bigg\}.\nonumber
 \end{align}

The event detection probability is the probability that at least one NM detects the target by intersecting it. The following result gives the target detection probability for a network of stationary NMs deployed as an MCP.

\begin{corollary}\label{Cor_statMCP}
 For a network of stationary NMs deployed as an MCP, the target detection probability is given by (for proof, see Appendix \ref{Ap_statMCP})
    \begin{align}
       \Tdp_\Mat^0&=1-\exp\left(-4\pi\lambda_{\rm p}\int_{0}^\infty\right.\left.\quad\left(1-\exp\left(-\frac{\bar{m}}{(4/3)\pi r^{3}}\mathcal{A}(a,r,x)\right)\right)x^2\diff x\right), \label{target-detection-MCP0}
    \end{align} where $\mathcal{A}(a,r,x)$ is the volume of intersection of $\ball{\origin,a}$ and $\ball{\x,r}$, and is given as \cite{kern1934}
    \begin{align*}
        &\mathcal{A}(a,r,x)=\begin{cases}
        \frac{4}{3}\pi\min(a^{3},r^{3})\quad \text{if } 0< x\leq|a-r|,\\
        \frac{\pi(a+r-x)^2}{12x}\times\qquad \text{if }|a-r|\leq x< a+r\\
\left(x^2+2x(a+r)-3(r-a)^2\right). 
        \end{cases}
    \end{align*} 
\end{corollary}
Similarly, we can obtain the detection probability of the target when the NMs are deployed as TCP, which is given in the following theorem.
\begin{corollary}\label{Cor_statTCP}
For a network of stationary NMs deployed as a TCP,  the target detection probability is (for proof, see Appendix \ref{Ap_statTCP})
      \begin{align}
        &\Tdp_\Tho^0=1-\exp\left(-4\pi\lambda_{\rm p}\int_0^\infty\left(1-\exp\left(-\int_{0}^a\right.\right.\right.\left.\left.\left.\frac{\sqrt{2}\mbar}{\sqrt{\pi\sigma^2}}\exp\left(-\frac{x^2+y^2}{2\sigma^{2}}\right)\sinh\left(\frac{xy}{\sigma^2}\right) \frac{y}{x}\diff y\right)\right)x^2\diff x\right).\label{target-detection-TCP0}
    \end{align}
\end{corollary}

\section{Target Detection Using a Single Cluster of NMs}
In this section, we consider a single cluster network of NMs for target detection where NMs are deployed in a single cluster. A single-cluster network captures the impact of proximity between the cluster center and the target, whereas a PCP-based deployment does not, since NMs are dispersed throughout the space. Let us denote the cluster center as $\x_\Pa$, which is assumed to be within a spherical region of radius $r_\Pa$ with PDF $g(\x)$. The centers of NMs are distributed around this cluster center with PDF $f(\y)$, independent of other NMs. We consider two scenarios of deployment of NMs. The first considers the centers of the NMs to be uniformly distributed within $\ball{\x_\Pa,r}$ around the cluster center.  In the second scenario, the centers of NMs are distributed according to a Gaussian distribution of variance $\sigma^2$ around the cluster center. These two depict the single cluster version of MCP and TCP respectively.
The boolean model for this single-cluster system at time $t=0$ can be described as
\begin{align}
    \Psi_\s = \bigcup\nolimits_{\y_j \in \mathcal{N}_\s} \left( \x_\Pa + \y_j + \ball{\origin, a} \right),
\end{align}
where $\mathcal{N}_\s$ represents the set of all daughter points (depicting NMs centers) within the cluster. 

We are interested in the detection event $\Event_\s$ that the point target at the origin is detected  by any of the
NMs, \ie
\begin{align}
    \Event_\s = \bigcup\nolimits_{\y_j \in \mathcal{N}} \left( \bp_{\Pa j}(t) \oplus \ball{\origin, a} \cap \origin \neq \emptyset \right),\label{event_single}
\end{align}
where $\bp_{\Pa j}(t)$ is the Brownian path of the $j$-th NM within the cluster, starting from $\x_\Pa + \y_j$ and moving according to Brownian motion with diffusion coefficient $D$. The single-cluster model is applicable in scenarios where NMs are concentrated in a localized region. By adjusting the parameters \(\x_\Pa\), \(r\), or \(\sigma^2\), the spatial characteristics of the NMs can be changed depending on the application requirements. Recall that the mathematical descriptions of the distributions of the NM centers are as per \eqref{Matdens} and \eqref{Thodens} for uniform and Gaussian distributions, respectively.
 \begin{theorem} \label{ThmSingle_actual}
  For NMs deployed in a single cluster, the probability of detecting a point target at the origin by any of the NMs within time $t$ is  (for proof, see Appendix \ref{Ap_sin_actual}).
\begin{align}
    \Tdp^\s(t)&=1- \int_{\rthree}\exp\left(-\mbar\int_{\rthree}\frac{a}{||\x+\y||}\right. \left.\erfc\left(\frac{||\x+\y||-a}{\sqrt{4Dt}}\right)f(\y)\diff \y\right)g(\x)\diff \x, \label{eq:singleC}
\end{align}where $g(\x)$ is the spatial PDF of the cluster center and $f(\y)$ is the PDF of the center of NMs around the cluster center.
 \end{theorem}
 Substituting the corresponding distributions as given in \eqref{Matdens} and \eqref{Thodens} into the above equation and by the rotational invariance property of the cluster center, we get the following corollaries.
 \begin{corollary}\label{Uniformcoro}
 For a network of a single cluster NMs with
 the cluster center uniformly distributed in
$\ball{\origin,R}$ and NMs uniformly distributed in $\ball{\x_\Pa,r}$ around the cluster center, the target detection probability within time $t$ is given as 
\begin{align}
    &\Tdp_\Mat^\s (t)=
1-\frac{3}{R^3} \int_{0}^R\exp\left(-\frac{3\mbar}{2 r^3}\int_{0}^r\int_0^\pi\frac{a}{\gamma(x,y,\theta)}\right. \left.\erfc\left(\frac{\gamma(x,y,\theta)-a}{\sqrt{4Dt}}\right)y^2\sin(\theta)\diff \theta\diff y\right)x^2\diff x.\label{Uniformtarget}
\end{align} 
\end{corollary}
 
 \begin{corollary}\label{Gausscoro}
  For a network of a single cluster NMs with
 the cluster center uniformly distributed in
$\ball{\origin,R}$ and NMs Gaussian distributed with variance $\sigma^2$, the target detection probability within time $t$ is given as 
     \begin{align}
&\Tdp_\Tho ^\s(t)=
1-\frac{3}{R^3}  \int_{0}^R\exp\left(-\frac{\mbar}{\sqrt{2\pi}\sigma^3}\int_{0}^\infty\int_0^\pi\frac{a}{\gamma(x,y,\theta)}\right.\left.\erfc\left(\frac{\gamma(x,y,\theta)-a}{\sqrt{4Dt}}\right)e^{-\frac{y^2}{2\sigma^2}}y^2\sin(\theta)\diff \theta\diff y\right) x^2\diff x.\label{Gausstarget}
\end{align} 
 \end{corollary}
Using a proof similar to that of Theorem \ref{ThApprox}, we can derive the approximate target detection probability for a single cluster scenario, where the cluster center is uniformly distributed, as presented in the following corollary.
\begin{corollary}\label{ThApproxSin}
The target detection probability of a single cluster of NMs when the cluster center is uniformly distributed in $\ball{\origin,R}$ and NMs are initially concentrated at the cluster center can be approximated as
\begin{align}
    \Tdp_\s(t)&\approx
1- \frac{3}{R^3}\int_{0}^R\exp\left(- \frac{\mbar a}{r}  \erfc\left(\frac{r-a}{\sqrt{4Dt}}\right)\right)\ r^2\diff r.
\end{align}
\end{corollary}


\section{Extension: Detection of a Spherical Target}\label{Sec:Extension}
Up until now, we have assumed that the target is a zero-dimensional point in space and have not taken its size into account. Let us consider the target as a spherical shape with a radius $a_\target$. The detection events of the spherical target by NMs in the cases of PCP and single cluster deployments are respectively given by
 \begin{align}
     \Event^\target &= \bigcup\nolimits_{\x_i \in \Phi_\Pa} \bigcup\nolimits_{\y_{ij} \in \mathcal{N}_i} \left( \bp_{ij}(t)\oplus \ball{\origin, a} \cap \ball{\origin, a_\target} \neq \emptyset \right)\nonumber\\
     &=\bigcup\nolimits_{\x_i \in \Phi_\Pa} \bigcup\nolimits_{\y_{ij} \in \mathcal{N}_i} \left( \bp_{ij}(t)\oplus \ball{\origin, a+a_\target} \cap \origin \neq \emptyset \right),\nonumber\\
    &\text{and }\Event^\target_\s = \bigcup\nolimits_{\y_j \in \mathcal{N}} \left( \bp_{\Pa j}(t) \oplus \ball{\origin, a+a_\target} \cap \origin \neq \emptyset \right).
\end{align}Therefore, when the target is modeled as a spherical object with radius $a_\target$, the event of detection of the target of radius $a_\target$ by NMs of radius $a$ is equivalent to the event of detection of a point target by NMs of radius $a+a_\target$. The detection probability expressions, bounds, and approximations derived for a system with a zero-dimensional point target remain valid with $a$ replaced with $a+a_\target$. 
\section{Numerical Results}\label{Sec:NumResults}
We validate our derived results using particle-based simulations and analyze the detection system's performance under different parameter variations. The particle simulations are designed to model the behavior of NMs in the 3D medium, with their movement governed by Brownian motion. The clustering processes describe the initial spatial distribution of the NMs at time $t=0$. The parent points themselves are distributed according to a homogeneous PPP in a bounded spherical region with a radius of $250\, \mu$m. The movement of NMs begins at time $t=0$ s, and follows a Brownian motion with a diffusion coefficient $D=100\,\mu$m$^2$/s. The particle simulation uses a time step of $\Delta t=0.001$ s. The target is considered to be detected if the spherical surface of any NM intersects with the target located at the origin. The detection probability is calculated by determining whether any NM is within a distance of $a$ from the origin during the simulation time. The results of these particle simulations are averaged over $10^4$  independent realizations to ensure statistical reliability.

\begin{figure}[ht]
	\centering
	\begin{minipage}{0.48\linewidth}
		\centering
		\includegraphics[width=\linewidth]{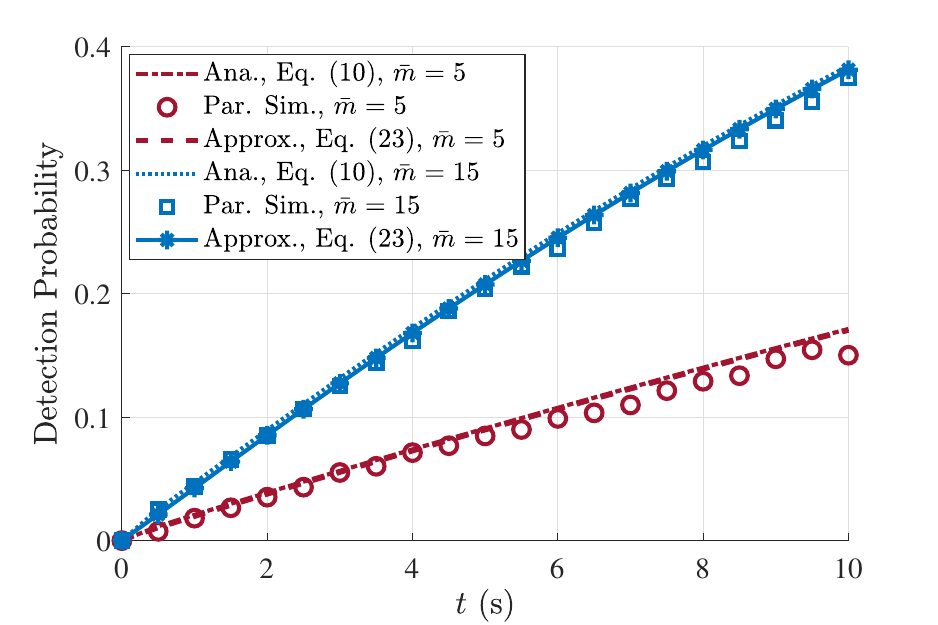}
		(a)
	\end{minipage}\hfill
	\begin{minipage}{0.48\linewidth}
		\centering
		\includegraphics[width=\linewidth]{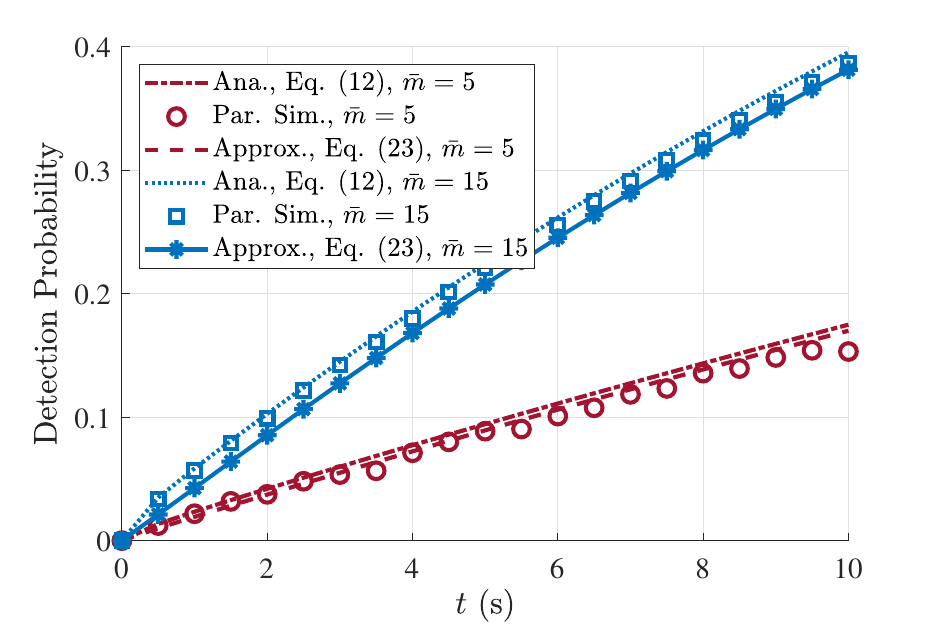}
		(b)
	\end{minipage}
	\caption{Variation of the detection probability over time $t$ for different values of $\mbar$ in (a) MCP and (b) TCP. The probability of detection over time of PCP deployed NMs shows a higher detection probability when $t$ and $\mbar$ are increased. The parameter values are $a=3\,\mu$m, $r=10\,\mu$m, $\sigma=10\,\mu$m, $\lambda_{\mathrm{Pa}}=1\times10^{-6}$ clusters/$\mu$m$^{3}$, $D=100\,\mu$m$^{2}$/s, and $\Delta t=10^{-3}$\,s.}
	\label{MCPTCPmbar}
\end{figure}

Fig. \ref{MCPTCPmbar} shows the variation in the probability of detection within time $t$ for a) MCP and b) TCP considering different values of $\mbar$. The derived equations are in good agreement with the particle-based simulation results, confirming the accuracy of the derived results. With an increase in $t$, the probability of target detection in time $t$ increases. With an increase in $\mbar$, the average number of NMs in a cluster increases, thus increasing the probability of detection of the target. Furthermore, the approximate equation for the detection probability for MCP is closer to the actual results compared to that for TCP. 
This can be attributed to the presence of distant NMs in TCP, due to the Gaussian spread, weakening the approximation $\|\x+\y\|\approx \|\x\|$ more compared to the strictly bounded clusters in MCP.

\begin{figure}[ht]
	\centering
	\begin{minipage}{0.48\linewidth}
		\centering
		\includegraphics[width=\linewidth]{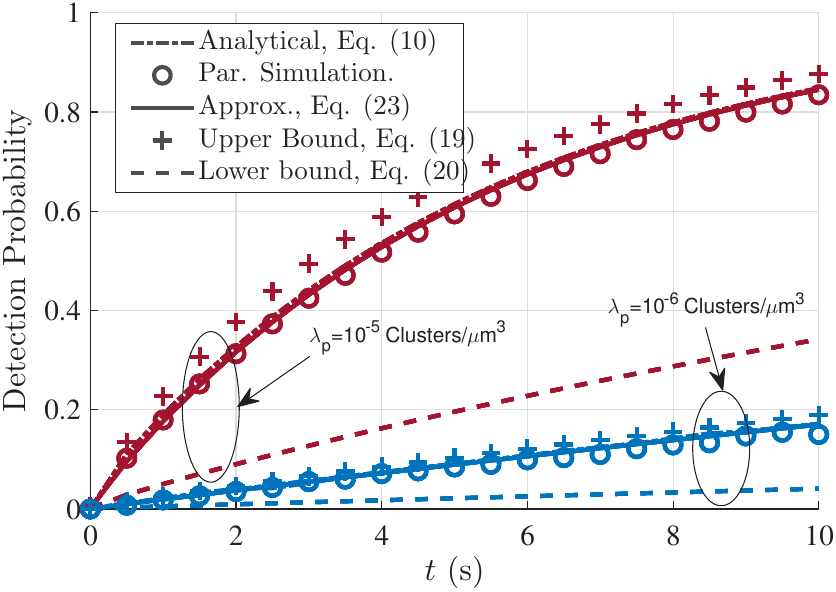}
		(a)
	\end{minipage}\hfill
	\begin{minipage}{0.48\linewidth}
		\centering
		\includegraphics[width=\linewidth]{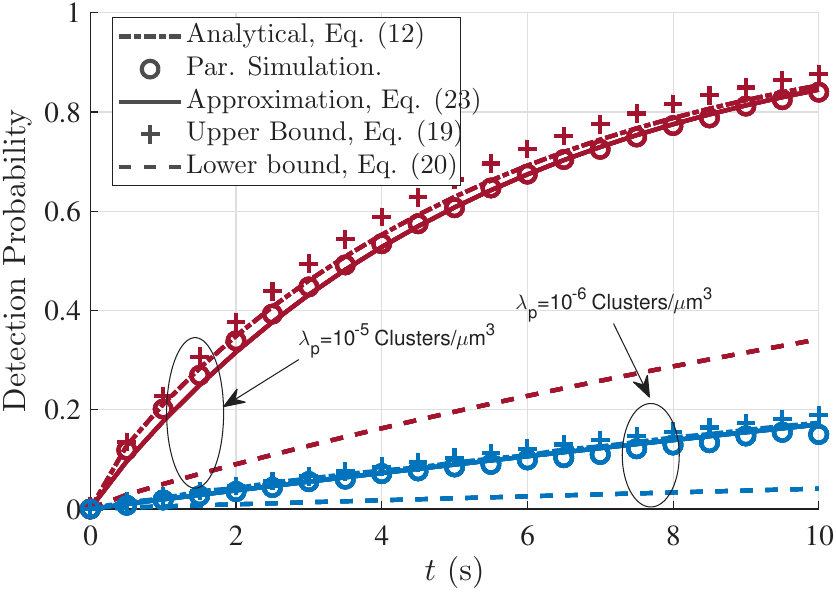}
		(b)
	\end{minipage}
	\caption{Variation of the detection probability over time $t$ for (a) MCP and (b) TCP. The probability of detection increases with cluster density. The parameter values are $a=3\,\mu$m, $r=10\,\mu$m, $\sigma=10\,\mu$m, $\mbar=5$, $D=100\,\mu$m$^{2}$/s, and $\Delta t=10^{-3}$\,s.}
	\label{MCPTCPlambda}
\end{figure}

Fig. \ref{MCPTCPlambda}  illustrates the variation in the probability of detection over time $t$ for (a) MCP and (b) TCP at different values of $\lambda_\Pa$. With an increase in $\lambda_\Pa$, the number of clusters per unit volume increases, thus increasing the number of NMs per unit volume and thereby the detection probability. The approximate equation for the detection probability is good, even if the cluster density is increased for the chosen parameters. Moreover, the numerical results of the exact equations are tightly bounded by the upper and lower bounds derived, ensuring their validity. In particular, the upper bound is significantly tighter than the lower bound. 

\begin{figure}[ht]
	\centering
	\begin{minipage}{0.48\linewidth}
		\centering
		\includegraphics[width=\linewidth]{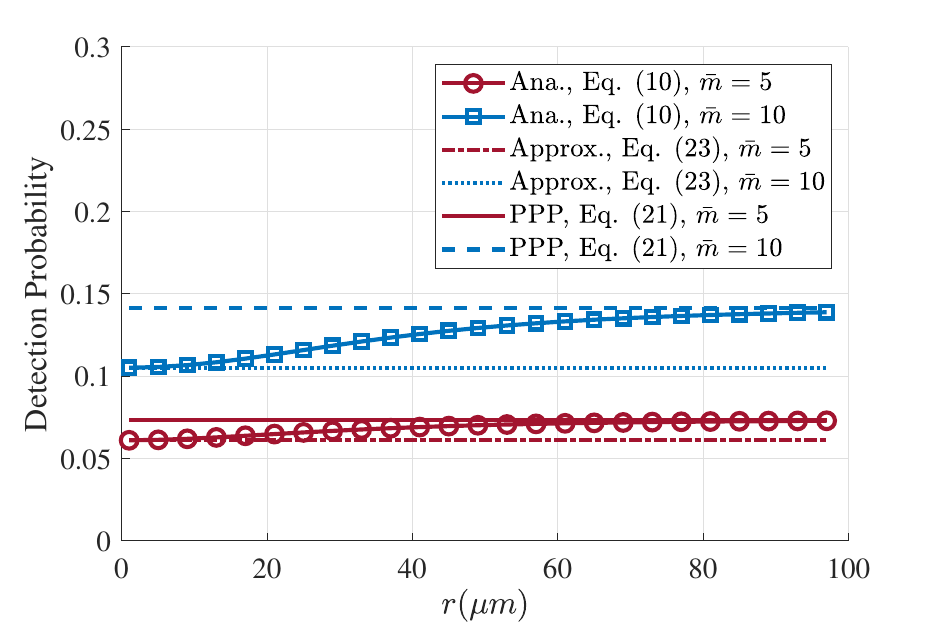}
		(a)
	\end{minipage}\hfill
	\begin{minipage}{0.48\linewidth}
		\centering
		\includegraphics[width=\linewidth]{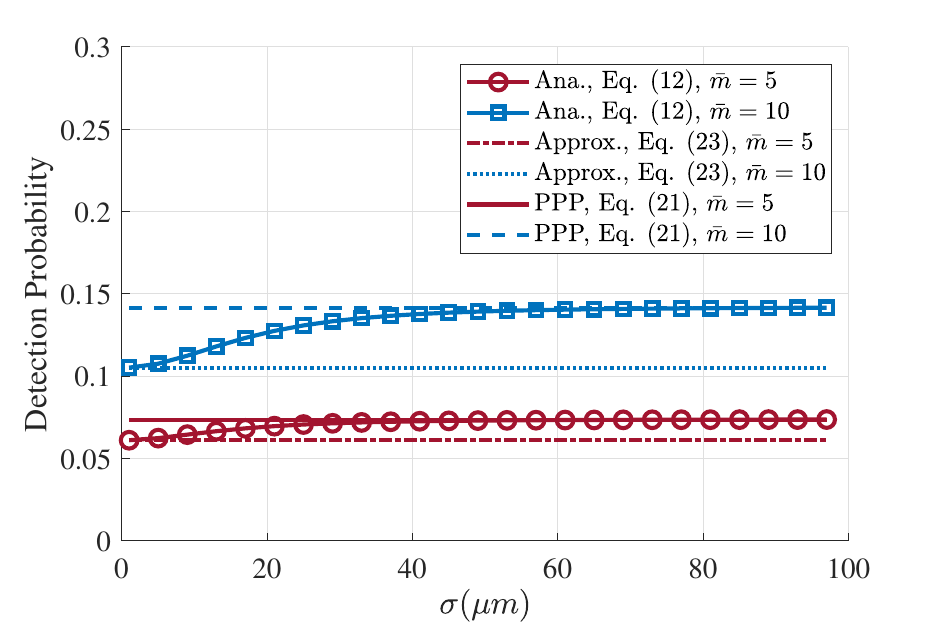}
		(b)
	\end{minipage}
	\caption{Variation of the detection probability over time $t$ versus (a) the cluster radius $r$ for MCP and (b) the spread parameter $\sigma$ for TCP. Increasing $r$ in MCP and $\sigma$ in TCP enhances the detection probability; however, the accuracy of the approximate expression degrades as NMs are distributed farther from the cluster centers. The parameter values are $a=4\,\mu$m, $\lambda_{\mathrm{Pa}}=5\times10^{-7}$ clusters/$\mu$m$^{3}$, $D=100\,\mu$m$^{2}$/s, $t=5$\,s, and $\Delta t=10^{-3}$\,s.}
	\label{MCPTCPrs}
\end{figure}

Fig. \ref{MCPTCPrs} shows the variation in the probability of detection with a) the cluster radius $r$ for MCP and b) the value $\sigma$ for TCP, which represents the spread of the spatial locations of the NMs. As the value of $r$ and $\sigma$ increase, the probability of detection increases. This is because increasing $r$ and $\sigma$ causes the initial distribution of NMs to spread more around the center of the cluster, thereby increasing the probability of detecting the target. As the value of $r$ and $\sigma$ increase, the actual detection probability and the approximate detection probability varies because $ \|\x_i + \y_{ij}\| \approx \|\mathbf{x}_i\| $ will no longer be satisfied. Furthermore, with increasing $r$ and $\sigma$, the detection performance of PCP approaches that of PPP. As $r$ and $\sigma$ become very large, PCP loses its clustering property and becomes uniformly distributed like PPP.
\begin{figure}[ht]
    \centering
    \includegraphics[width=0.45\linewidth]{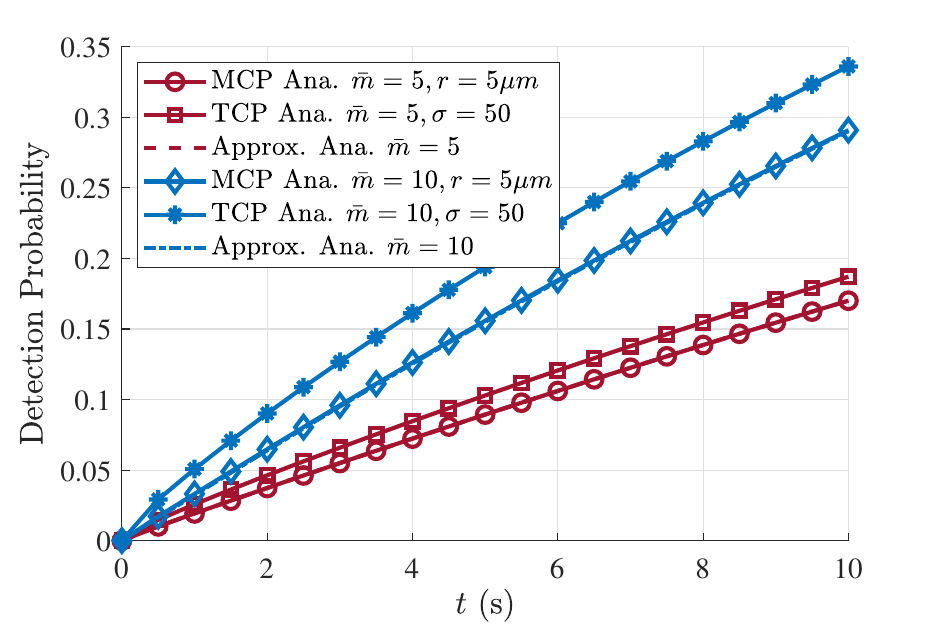}
    \caption{A comparison of target detection probabilities for MCP  and TCP distributed NMs shows that both methods have similar detection rates when NMs are near cluster centers. However, TCP performs better at higher values of $\sigma$. The parameters are \( a = 3 \, \mu m \), \( \lambda_\Pa = 1 \times 10^{-6} \, \text{Clusters}/\mu m^{3} \), and \( D = 100 \, \mu m^2/s \).}
    \label{fig:TCP_MCPComp}
\end{figure}

Fig. \ref{fig:TCP_MCPComp} shows the comparison of the target detection probability of the network of NM deployed as MCP and TCP. Although Figs. \ref{MCPTCPmbar}, \ref{MCPTCPlambda}, and \ref{MCPTCPrs} suggest that NMs deployed by MCP and TCP have similar detection probabilities, this is not always the case. Their similar performance here is due to the chosen parameters, which initially position the NMs close to the cluster center. However, Fig. \ref{fig:TCP_MCPComp} shows that TCP can perform better when the value $\sigma$ is higher keeping $\mbar$ the same. Similarly, we can verify that MCP can also perform better for large values of cluster radius.
 
\begin{figure}
    \centering
    \includegraphics[width=0.45\linewidth]{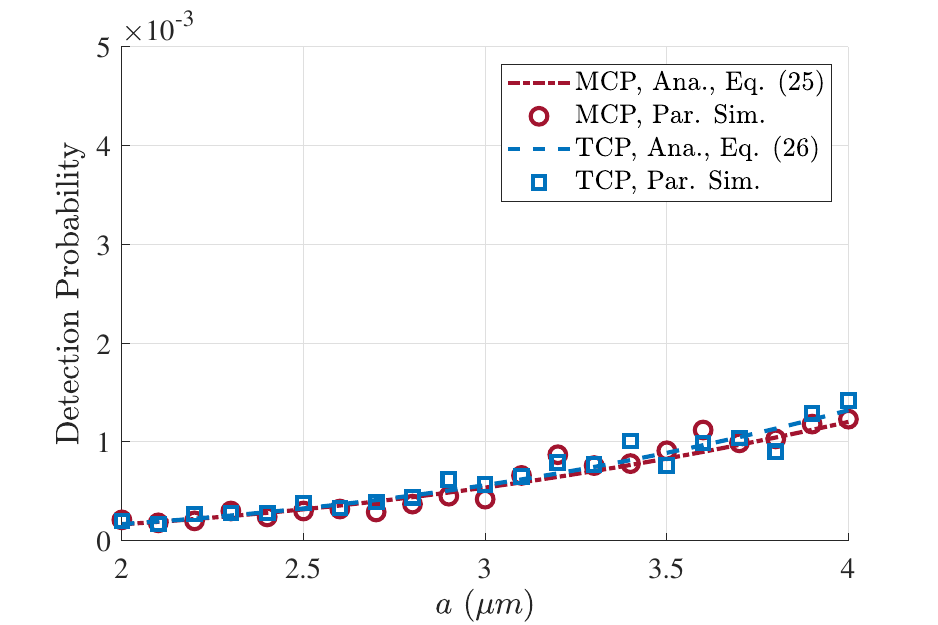}
     \caption{Target detection probability at time $t=0$ versus NM radius for MCP and TCP deployed NMs. Increasing NM size enhances detection probability due to a larger surface area and closer proximity to the target.  The parameter values are $r=10\, \mu$m, $ \sigma=10\,\mu$m, $a= 3\, \mu$m, $ \lambda_{\Pa}=1\times 10^{-6}\text{ Clusters}/\mu$m$^3$,  and   $D= 100\, \mu $m$^2/s$.}
    \label{fig:Stat}
\end{figure}
{\color{black} Fig. \ref{fig:Stat} shows the detection probability at $t=0$ with respect to receiver radius $a$ for NMs distributed as MCP and PCP, given by \eqref{target-detection-MCP0} and  \eqref{target-detection-TCP0}, respectively. Note that $t=0$ refers to the time when NMs are at their initial location as deployed and they have not yet moved. We can first observe that the analytical curves closely match particle-based simulations, validating the derived expressions. Because $r$ and $\sigma$ have comparable values, the detection performance of both deployments is very similar. The detection performance of TCP is slightly better compared to that of MCP due to its higher spatial spread. Furthermore, as the NM size $a$ increases, the detection probability increases due to the increase in the volume swept by the NMs of a cluster.}

{\color{black}Till now, we analyzed the detection performance in PCP deployments with clusters spread everywhere in an unbounded environment. We now focus on the single cluster scenario, where a single cluster of a finite number of NMs is deployed around a cluster center. The NMs are assumed to be located either uniformly within a spherical region of radius $R$  or following a Gaussian distribution with variance $\sigma^2$. }
\begin{figure}[ht]
    \centering
    \includegraphics[width=0.45\linewidth]{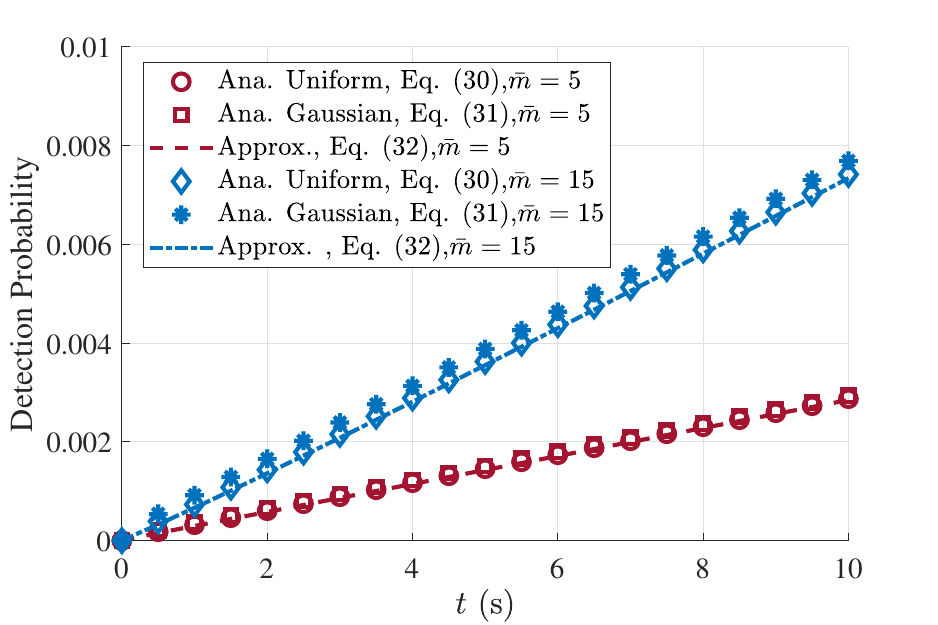}
    \caption{Detection probability  over time $ t $ of single cluster scenario with $ \mbar = 5 $  and $ \mbar = 15 $. Detection probability increases as the mean number of NMs per cluster rises.  The parameter values are $r=10\ \mu$m$,\ \sigma=10\,\mu$m, $ a= 3\ \mu \text{m}, \text{ and } D= 100\,\mu \text{m}^2/s$.}
    \label{fig:SinCluster}
\end{figure}

Fig. \ref{fig:SinCluster} illustrates the variation in detection probability over time $t$ for the single cluster scenario. The center of the cluster is uniformly distributed in a spherical region of radius $R=200\,\mu $m. The detection probabilities of uniformly distributed and Gaussian distributed NMs in the cluster are nearly identical for the given parameters. As the mean number of NMs per cluster increases, the probability of detection concurrently increases. 

\begin{figure}[ht]
    \centering
    \includegraphics[width=0.45\linewidth]{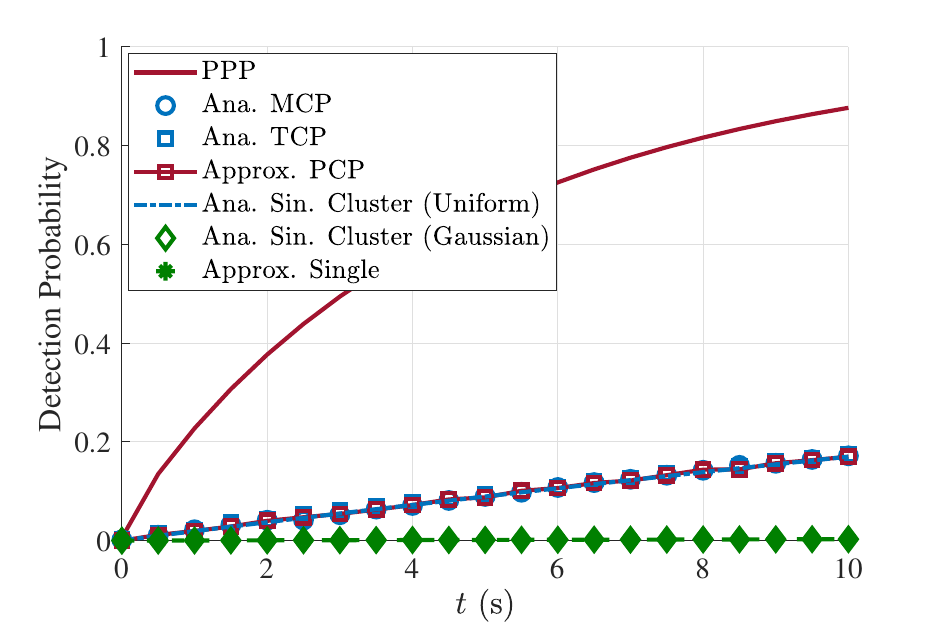}
    \caption{Comparison of detection probabilities of PPP, PCP and single cluster network of NMs. The parameter values are  $\mbar=5,\ \lambda_\Pa=1\times 10^{-5}\text{Clusters/}\mu \text{m}^3,\  r=10\ \mu \text{m}, \, \sigma=10\,\mu$m, $ a= 3\,\mu \text{m}, \text{ and } D= 100\,\mu \text{m}^2/s$.}
    \label{fig:comparison}
\end{figure}
Fig.~\ref{fig:comparison} shows the comparison of NM networks with different spatial arrangements like PPP, PCP (MCP and TCP), and single cluster scenario. The PCP performs significantly better than the single-cluster scenario in terms of detection probability because of the multiple clusters. However, PPP with effective density $\lambda_\Pa\mbar$ outperforms PCP because uniformly distributed NMs in PPP provide greater coverage across the medium, reducing the likelihood of detection gaps. This comparison highlights the trade-off between localized clustering and uniform spatial distribution in the design of NM networks for efficient target detection.

{\color{black}The performance gap between the PCP and PPP deployments, as seen in Fig.~\ref{fig:comparison}, aligns with the analysis in Section~\ref{PerfComp}. One way to reduce the performance loss is by using molecules with a higher diffusion coefficient $D$.  A higher $D$ enhances the mobility of NMs, allowing them to explore a larger volume over time $t$. This increased diffusion reduces the impact of clustering by enabling NMs to spread more uniformly, effectively transitioning the PCP's spatial distribution toward a PPP-like coverage as $t$ increases, thereby improving detection probability and narrowing the performance gap.	
	
		The numerical results provide insights into improving detection probability in PCP deployments by adjusting key parameters, as illustrated in Figs.~\ref{MCPTCPmbar}, \ref{MCPTCPlambda}, and \ref {fig:TCP_MCPComp}. For instance, increasing \(\mbar\) (more NMs per cluster, as in Fig.~\ref{MCPTCPmbar}) or \(\lambda_\Pa\) (more clusters, as shown in Fig.~\ref{MCPTCPlambda})  enhances detection probability by increasing the overall NM density. Similarly, a larger cluster radius \(r\) or standard deviation \(\sigma\) (Fig.~\ref{fig:TCP_MCPComp}) allows NMs to spread over a wider area, improving coverage. 	}

{\color{black}We can observe that there are significant differences in detection probability magnitudes across presented results, which is due to differences in system configurations and parameter regimes used in each case. 	Figures~\ref{MCPTCPmbar}, \ref{MCPTCPlambda}, \ref{MCPTCPrs}, \ref{fig:TCP_MCPComp}, and \ref{fig:Stat} correspond to large-scale PCP systems with multiple clusters, which naturally yield higher detection probabilities due to the larger number of NMs. 	In contrast, Fig.~\ref{fig:SinCluster} represents a single-cluster deployment, where the limited number of NMs and bounded spatial extent lead to much smaller detection probabilities. 	Furthermore, Fig.~\ref{fig:Stat} shows the detection probability at time $t=0$, thus giving lower values compared to time-evolving results.}

Although a PPP-based deployment offers better detection probability, practical constraints could force clustered deployments.  Diffusion over time can cause PCP distributions to resemble PPP, although this transition causes detection delays.  Therefore, adjusting cluster characteristics or using diffusion-enhancing methods can help minimize performance loss when clustered deployments are inevitable, ensuring more effective target detection in NM networks.

\section{Conclusion}\label{Conclusions}
We have presented an analytical framework for evaluating the
target detection probability of a diffusion-assisted molecular communication-based network with NMs with clustered initial deployment. We then analyzed the target detection probability for these molecules and compared it with NMs that do not exhibit clustering. We derived analytical expressions and bounds for the target detection probability by considering the initial deployment of NMs using PCP (MCP and TCP). The derived equations were validated using extensive particle-based simulations. We also investigated a single-cluster situation in which NMs were either uniformly or Gaussian-distributed inside the cluster.

Moreover, we investigated PCP-based deployments against PPP of equivalent NM density and showed that clustered deployment can result in reduced performance. {\color{black}This performance loss is due to the fact that the  NMs are concentrated around cluster centers, leading to coverage gaps in regions far from these centers. To improve the detection performance, parameters such as the diffusion coefficient \(D\), cluster density \(\lambda_\Pa\), or cluster spread (\(r\) for MCP, \(\sigma\) for TCP) can be adjusted. For instance, a higher \(D\) enables NMs to cover larger areas via Brownian motion, reducing gaps over time, while increasing \(\mbar\) or \(\lambda_\Pa\)  boosts the effective NM density. Larger \(r\) or \(\sigma\) spreads NMs, approaching PPP-like performance but diminishing clustering advantages.} {\color{black}These results underscore the need for deployment techniques to optimize molecular communication systems, particularly for applications such as targeted drug delivery to cancer cells, early detection of disease biomarkers, and monitoring of chemical pollutants in biological or environmental systems.}

\appendices
\section{}\label{Ap_actual}
The event $\Event$ in \eqref{event1} can also be written as,
 \begin{align}
     \Event = \left(\bigcap\nolimits_{i \in \mathbb{N}} \bigcap\nolimits_{j \in \mathbb{N}} \left( \bp_{ij}(t)\oplus \ball{\origin, a} \cap \origin = \emptyset \right)\right)^\Comp.\label{event2}
 \end{align}Now, let’s define an indicator function  for this event as
\begin{align}
     &\Ind_\Event=1-\prod\nolimits_{\x_i\in\Phi_\Pa}\prod\nolimits_{\y_{ij}\in\mathcal{N}_i}\mathds{1}\left( \bp_{ij}(t)\oplus \ball{\origin, a} \cap \origin = \emptyset \right).\nonumber
 \end{align}
 The probability of the occurrence of event $\Event$ is given by
\begin{align}
\Prob\left[\Ind_\Event\right]&=1-\Expect_{\x,\y,\bp}\left[\prod\nolimits_{\x_i\in\Phi_\Pa}\prod\nolimits_{\y_{ij}\in\mathcal{N}_i}\right.\nonumber\\
& \left.\mathds{1}\left( \bp_{ij}(t)\oplus \ball{\origin, a} \cap \origin = \emptyset \right)\right]\nonumber\\
&\stackrel{(a)}=1-\Expect_{\x}\left[\prod\nolimits_{\x_i\in\Phi_\Pa}G(\x_i)\right]\nonumber\\
&\stackrel{(b)}=1-\exp\left(-\lambda_\Pa \int_{\rthree}\left(1-G(\x)\right)\diff \x\right),\label{eqnbase}
\end{align}where step \((a)\) is obtained by replacing 
\begin{align}    &G(\x_i)=\Expect_{\y,\bps}\left[\prod\nolimits_{\y_{ij}\in\mathcal{N}_i}\right.\left.\vphantom{\frac{}{}}\mathds{1}\left( \left( \x_i + \y_{ij} + \bps_{ij}(t) \right) \oplus \ball{\origin, a} \cap \origin = \emptyset \right)\right],\label{eqn:gx1}
\end{align}and step \((b)\) is obtained using the PGFL \cite{haenggi2012} of PPP. Simplifying $G(\x)$  gives,
\begin{align}
    &G(\x)\stackrel{(a)}=\exp\left(-\mbar\int_{\rthree}\left(1-\right.\right.\left.\Expect\left[\mathds{1}\left( \left( \x + \y + \bps(t) \right) \oplus \ball{\origin, a} \cap \origin = \emptyset \right)\right]\right)f(\y)\diff \y\bigg)\nonumber\\
    &\stackrel{(b)}=\exp\left(-\mbar\int_{\rthree}\right.\Prob\left[\left( \x + \y + \bps(t) \right) \oplus \ball{\origin, a} \cap \origin \neq \emptyset \right]f(\y)\diff \y\bigg)\nonumber\\
    &\stackrel{(c)}=\exp\left(-\mbar\int_{\rthree}\Prob\left[\bps(t)\cap \ball{\x+\y,a} \neq\phi\right]f(\y)\diff \y\right)\nonumber\\
   &\stackrel{(d)}= \exp\left(-\mbar\int_{\rthree}\frac{a}{||\x+\y||}\erfc\left(\frac{||\x+\y||-a}{\sqrt{4Dt}}\right)f(\y)\diff \y\right),\nonumber
\end{align}where step $(a)$ is derived using the PGFL of daughter PP, step $(b)$ is obtained by simplifying the step $(a)$, step $(c)$ is due to definition of the Minkowski's sum and finally step $(d)$ is obtained from \cite{yilmaz2014b}, 
\begin{align}
    \Prob\left[\bps(t)\cap \ball{\x,a} \neq\phi\right]=\frac{a}{||\x||}\erfc\left(\frac{||\x||-a}{\sqrt{4Dt}}\right).\label{eq:hp3d}
\end{align}
Therefore, the probability of the event \(\Event\) is 
\begin{align}
&\Prob\left[\Ind_\Event\right]=
1-\exp\left(-\lambda_\Pa \int_{\rthree}\left(1-\exp\left(-\mbar\int_{\rthree}\right.\right.\right. \left.\left.\frac{a}{||\x+\y||}\erfc\left(\frac{||\x+\y||-a}{\sqrt{4Dt}}\right)f(\y)\diff \y\bigg)\right)\diff \x\right).\label{hitproap}
\end{align}The detection probability of a target molecule by MCP and TCP distributed NMs can be determined by substituting \eqref{Matdens} and \eqref{Thodens} into \eqref{hitproap}.
\section{}\label{Ap_Thmean}
The clusters that detect the target are given as $\{\x_i: \text{there exist a } \y_{ij}\in \Psi_t \text{ with }  \Event_{ij} \}$, which can be seen as independent thinning of $\Phi_\Pa$. Therefore, from the independent thinning theorem, these detecting clusters form a PPP. Hence, the number of clusters is  Poisson RV with mean
\begin{align}
  N_\Psi&= \Expect\left[\sum_{\x_i} \mathds{1}\left(\cup_{\y_{ij}} \Event_{ij} \right)\right].
\end{align}
Applying Campbell's theorem\cite{andrews2023a},
\begin{align}
  N_\Psi&=\lambda_\Pa\int \Prob\left[\cup_{\y_{j}} \Event_{j}\right] \diff \x,  
\end{align}
where $\Event_{j}=\{ (\x+\y_j+\bps_j(t))\oplus \ball{\origin, a} \cap \{\origin\} \neq \emptyset \}$.
Using \eqref{eqn:gx1} and $\Event_{j}$, we get
\begin{align}
   1-G(\x)= 1-\Prob\left[\cap_{\y_{j}} \Event_{j}^\Comp\right]= \Prob\left[\cup_{\y_{j}} \Event_{j}\right].
\end{align}
Hence, we get the mean as
\begin{align}
    N_\Psi&= \lambda_\Pa V(t),\\
   \text{where,}\quad V(t)&=\int_{\rthree}\left(1-G(\x)\right)\diff \x.\label{Vdot}
\end{align}From \eqref{eqn:gx1}, $G(\x)$ can be rewritten as
\begin{align}
&\stackrel{(a)}=\Expect_{\y,\bps}\left[\prod\nolimits_{\y_{j}\in\mathcal{N}_\origin}\mathds{1}\left( \origin\notin  \x + \y_{j} + \bps_{j}(t)  \oplus \ball{\origin, a} \right)\right]\nonumber\\
&\stackrel{(b)}=\Expect_{\y,\bps}\left[\prod\nolimits_{\y_{j}\in\mathcal{N}_\origin}\mathds{1}\left( -\x \notin  \y_{j} + \bps_{j}(t)  \oplus \ball{\origin, a} \right)\right]\nonumber\\
&\stackrel{(c)}=\Expect_{\y,\bps}\left[\mathds{1}\left( \bigcap\nolimits_{\y_{j}\in\mathcal{N}_\origin}-\x \notin  \y_{j} + \bps_{j}(t)  \oplus \ball{\origin, a} \right)\right],
\end{align}where step \((a)\) is obtained using a simple mathematical manipulation. Step \((b)\) follows from the property that the underlying PP is rotation-invariant, ensuring that the transformation does not affect the distribution of the points. Step $(c)$ is derived from the property that the product of indicator functions can be expressed as a single indicator function applied to the intersection of the corresponding events, \ie  
\[
\prod\nolimits_{j} \mathds{1}(\Event_j) = \mathds{1}\left(\bigcap\nolimits_{j} \Event_j\right).
\]  
This allows rewriting the expectation in terms of a single indicator function. Now, using the identity $\left(\bigcap_iA_i\right)^\Comp=\bigcup_i A_i^\Comp$,
\begin{align}
&G(\x)=1-\Expect_{\mathcal{N}_\origin,\bps}\left[\mathds{1}\left( -\x \in  \bigcup\nolimits_{\y_{j}\in\mathcal{N}_\origin} \y_{j} + \bps_{j}(t)  \oplus \ball{\origin, a} \right)\right]\label{eq:gxn2}.
\end{align}
Substituting \eqref{eq:gxn2} in \eqref{Vdot}, we get $V(t)$
\begin{align}
&=\int_{\rthree}\!\!\Expect_{\mathcal{N}_\origin,\bps}\left[\mathds{1}\left(\! -\x \in \bigcup\nolimits_{\y_{j}\in\mathcal{N}_\origin}\y_{ij} + \bps_{j}(t)  \oplus \ball{\origin, a} \right)\right]\diff \x\nonumber\\    
&=\Expect_{\mathcal{N}_\origin,\bps}\left[\int_{\rthree}\mathds{1}\left( -\x \in  \bigcup\nolimits_{\y_{j}\in\mathcal{N}_\origin}\y_{j} + \bps_{j}(t)  \oplus \ball{\origin, a} \right)\diff \x\right]\nonumber\\
&=\Expect_{\mathcal{N}_\origin,\bps}\left[\bigm|  \bigcup\nolimits_{\y_{j}\in\mathcal{N}_\origin} \y_{j} + \bps_{j}(t)  \oplus \ball{\origin, a}\bigm| \right].
\end{align}

Let us consider $\eta(t)=\bigm|  \bigcup\nolimits_{\y_{j}\in\mathcal{N}_\origin} \y_{j} + \bps_{j}(t)  \oplus \ball{\origin, a}\bigm|$, and let $V(t)$  denote the average volume of the combined region covered by the NMs of an arbitrary cluster. Then,
\begin{align}
	V(t)&=\int_{\rthree}\mathbb{P}[\z \in \eta(t)]\diff \z
	\nonumber\\
	&=\int_{\rthree}1-\mathbb{P}\left[\bigcap_{\y_{j}}\left(\z \notin \y_{j} + \bps_{j}(t)  \oplus \ball{\origin, a}\right)\right]\diff \z\nonumber\\
	&=\int_{\rthree}1-\mathbb{E}\left[\prod_{\y_{j}}\left(\z \notin \y_{j} + \bps_{j}(t)  \oplus \ball{\origin, a}\right)\right]\diff \z
	\end{align}
Applying the PGFL of PPP in the above equation gives,
\begin{align}
	V(t)
	&=\int_{\rthree}1-\exp\left(\mathbb{P}\left[\bps(t)\cap\ball{\z-\y,a}\right]f(\y)\diff \y \right)\diff \z
\end{align}
Now, using \eqref{eq:hp3d} and solving further gives \eqref{eq:V(t)}.

 Using the expression of $V(t)$ we get the average number of NMs distributed spatially as PCP detecting the target at the origin as $\lambda_\Pa\times V(t)$.

\section{}\label{VolineqAp}

From \cite[Eq. 4.30]{albeverio1998}, 
    \begin{align}
\bigm|   \bps(t)  \oplus \ball{\origin, a}\bigm| &\leq \bigm|  \bigcup\nolimits_{\y_{j}\in\mathcal{N}_\origin} \y_{j} + \bps_{j}(t)  \oplus \ball{\origin, a}\bigm|   \leq m\bigm|  \bps(t)  \oplus \ball{\origin, a}\bigm| ,\label{volineqap1}
\end{align}
where $m$ is the number of NMs in the typical cluster, which is Poisson distributed with mean $\mbar$. Let $\mathsf{A}$ be the event that there exists at least one NM in the typical cluster. Using the law of total expectation,
\begin{align}
    &\Expect_{\bps}\left[\bigm|   \bps(t)  \oplus \ball{\origin, a}\bigm| \right]=\Expect_{\bps}\left[\bigm|   \bps(t)  \oplus \ball{\origin, a}\bigm| \mid \mathsf{A}\right]\times\left(1-\exp(-\mbar)\right),\label{volLB}\\
    &\text{and}\nonumber\\
    &\Expect_{\bps}\left[m \bigm|  \bps(t)  \oplus \ball{\origin, a}\bigm| \right]=\Expect\left[{m}/{\mathsf{A}}\right]\times\left(1-\exp(-\mbar)\right)\times\Expect_{\bps}\left[\bigm|   \bps(t)  \oplus \ball{\origin, a}\bigm| \mid \mathsf{A}\right]\nonumber\\
    &=\mbar \Expect_{\bps}\left[\bigm|   \bps(t)  \oplus \ball{\origin, a}\bigm| \mid \mathsf{A}\right].\label{volUB}
\end{align}
Note that, $\Expect_{\bps}\left[\bigm|   \bps(t)  \oplus \ball{\origin, a}\bigm| \mid \mathsf{A}\right]$ denotes the average volume covered by a spherical particle of radius $a$ within time $t$ under Brownian motion \cite{berezhkovskii1989}.
Now, taking the expectation of \eqref{volineqap1} and substituting \eqref{volLB} and \eqref{volUB} into it gives \eqref{volineqex}.
\section{}\label{Ap_statMCP}
The indicator function $\Ind_{\Event_{\origin}}$ for the event $\Event_{\ob}$ is defined as
\begin{align*}
     &=1-\prod\nolimits_{\x_i\in\Phi_\Pa}\prod\nolimits_{\y_{ij}\in\mathcal{N}_i}\mathds{1}\left( \x_i + \y_{ij}\oplus \ball{\origin, a} \cap \origin = \emptyset \right).
 \end{align*}
 The probability of the occurrence of the event $\Event$ is
\begin{align}
&\Prob\left[\Ind_{\Event_\origin}\right] \stackrel{(a)}{=} 1 - \Expect_{\x,\y}\left[\prod\nolimits_{\x_i\in\Phi_\Pa}\prod\nolimits_{\y_{ij}\in\mathcal{N}_i}\right.\left.\vphantom{\frac{}{}}\Ind\left( \x_i + \y_{ij} \oplus \ball{\origin, a} \cap \origin = \emptyset \right)\right]\nonumber\\
&\stackrel{(b)}{=} 1 - \Expect_{\Phi_{\Pa}}\left[\prod\nolimits_{\x_i\in\Phi_\Pa}\prod\nolimits_{\y_{ij}\in\mathcal{N}_i}\Ind(\y_{ij} \notin \ball{-\x_{i},a})\right]\nonumber\\
&\stackrel{(c)}{=} 1 - \Expect_{\Phi_{\Pa}}\left[\prod\nolimits_{\x_i\in\Phi_\Pa}\exp\left(\int_{\y \in \ball{\origin,a}}\frac{\bar{m}}{(4/3)\pi r^{3}}\right.\right.\left.\left.\vphantom{\frac{}{}}\left(1 - \Ind(\y \notin \ball{-\x_{i},a})\right)\diff \y\right)\right]\nonumber\\
&\stackrel{(d)}{=} 1 - \Expect_{\Phi_{\Pa}}\left[\prod\nolimits_{\x_i\in\Phi_\Pa}\exp\left(-\frac{\bar{m}}{(4/3)\pi r^{3}}|\ball{\origin,a}\cap\ball{\x_{i},r}|\right)\right]
\end{align}Here step $(a)$ is due to the definition of the probability of event $\Event_\origin$. Step $(b)$ uses the definition of Minkowski sum  \cite{chiu2013stochastic}. In step $(c)$, we apply the PGFL of PPP within a ball of radius $a$. Step $(d)$ employs the stationarity property of the MCP. Using the PGFL of a three-dimensional PPP and simplifying further completes the proof of the theorem.
\section{}\label{Ap_statTCP}
 The event detection probability is 
    \begin{align}
&\Prob\left[\Ind_{\Event_\origin}\right]\stackrel{(a)}=1-\Expect_{\x,\y}\left[\prod\nolimits_{\x_i\in\Phi_\Pa}\prod\nolimits_{\y_{ij}\in\mathcal{N}_i}\right.\nonumber\\
   &\qquad \left.\mathds{1}\left( \x_i + \y_{ij}\oplus \ball{\origin, a} \cap \origin = \emptyset \right)\right]\nonumber\\
   &\stackrel{(b)}=1-\Expect_{\Phi_{\Pa}}\left[\prod\nolimits_{\x_i\in\Phi_\Pa}\prod\nolimits_{\y_{ij}\in\mathcal{N}_i}\mathds{1}(\y_{ij}\notin \ball{-\x_{i},a})\right]\nonumber\\
   &\stackrel{(c)}=1-\Expect_{\Phi_{\Pa}}\left[\prod\nolimits_{\x_i\in\Phi_\Pa}\right.\left.\quad\exp\left(-\int_{\y \in \ball{-\x_{i},a}}\frac{\bar{m}}{(2\pi\sigma^{2})^{3/2}}\exp\left(-\frac{\|\y\|^{2}}{2\sigma^{2}}\right)\diff \y\right)\right],\nonumber
\end{align}where step $(a)$ follows from the definition of the probability of event $\Event_\origin$, step $(b)$ applies the definition of the Minkowski sum, and finally, step $(c)$ uses the PGFL of the daughter PP. Replacing $\y=\x_{i}+{\bf t}\implies \diff\y= \diff {\bf t}$, we get
\begin{align*}
     \Prob\left[\Ind_{\Event_\origin}\right]  &=1-\Expect_{\Phi_{\Pa}}\left[\prod\nolimits_{\x_i\in\Phi_\Pa}\exp\left(-\int_{{\bf t} \in \ball{\origin,a}}\frac{\bar{m}}{(2\pi\sigma^{2})^{3/2}}\right.\right.\left.\left.\exp\left(-\frac{\|\x_{i}+{\bf t}\|^{2}}{2\sigma^{2}}\right)\diff {\bf t}\right)\right].
\end{align*}Further, applying the PGFL of PPP completes the proof.

\section{}\label{Ap_sin_actual}
The event $\Event_\s$ given in \eqref{event_single} can be written as,
 \begin{align}
     \Event_\s = \left(\bigcap\nolimits_{j \in \mathbb{N}_\s} \left( \bp_{\Pa j}(t)\oplus \ball{\origin, a} \cap \origin = \emptyset \right)\right)^\Comp.\label{events2}
 \end{align}The indicator function for this event can be written as
\begin{align}
     \Ind_{\Event_\s}=1-\prod\nolimits_{\y_j\in\mathcal{N}_i}\mathds{1}\left( \bp_{\Pa j}(t)\oplus \ball{\origin, a} \cap \origin = \emptyset \right).
 \end{align}
The probability of event \(\Event_\s\) can be determined as
\begin{align}
\Prob\left[\Ind_{\Event_\s}\right]&=1-\Expect_{\x,\y,\bp}\left[\prod\nolimits_{\y_j\in\mathcal{N}_\s}\mathds{1}\left( \bp_{\Pa j}(t)\oplus \ball{\origin, a} \cap \origin = \emptyset \right)\right]\nonumber\\
&\stackrel{(a)}{=}1-\Expect_{\x}\left[\exp\left(-\mbar\int_{\rthree}\left(1-\right.\right.\right.\left.\Expect\left[\mathds{1}\left( \left( \x + \y + \bps(t) \right) \oplus \ball{\origin, a} \cap \origin = \emptyset \right)\right]f(\y)\diff \y\bigg)\right]\nonumber ,
\end{align}
where (a) is obtained by applying PGFL of PPP.
Further simplifying the equation results in  
\begin{align}
&\Prob\left[\Ind_{\Event_\s}\right]
=1-\Expect_{\x}\left[\exp\left(-\mbar\int_{\rthree}\Prob\left[\bps(t)\cap \ball{\x+\y,a} \neq\phi\right]\right.\right.\left.\left.\times f(\y)\diff \y\right)\right]\nonumber\\ 
&\stackrel{(a)}{=}1-\Expect_{\x}\left[\exp\left(-\mbar\int_{\rthree}\frac{a}{||\x+\y||}\times\right.\right.\left.\erfc\left(\frac{||\x+\y||-a}{\sqrt{4Dt}}\right)f(\y)\diff \y\bigg)\right],
\end{align}
where (a) is obtained by using \eqref{eq:hp3d}. Now, taking the expectation with respect to $\x$ gives,
\begin{align}
\Prob\left[\Ind_{\Event_\s}\right]
&=1- \int_{\rthree}\exp\left(-\mbar\int_{\rthree}\frac{a}{||\x+\y||}\right. \left.\erfc\left(\frac{||\x+\y||-a}{\sqrt{4Dt}}\right)f(\y)\diff \y\right)g(\x)\diff \x,
\end{align}
where $g(\x)$ is the spatial distribution of the cluster center.

\ifCLASSOPTIONcaptionsoff
  \newpage
\fi
\bibliographystyle{IEEEtran}
\bibliography{ClusteredTD}

\end{document}

%% file: define.tex
\newcommand{\target}{\mathsf{t}}
\newcommand{\stat}{\mathsf{s}}
\newcommand{\Tx}{\mathsf{T}}
\newcommand{\Rx}{\mathsf{R}}
\newcommand{\Ind}{\mathds{1}}
\newcommand{\Bt}{\mathsf{B_t}}
\newcommand{\Rel}{\mathsf{Re}}
\newcommand{\Comp}{\mathrm{c}}
\newcommand{\Erf}{\mathrm{erf}}
\newcommand{\Erfc}{\mathrm{erfc}}
\newcommand{\Mat}{\mathsf{M}}
\newcommand{\Pcp}{\mathsf{PCP}}
\newcommand{\PPP}{\mathsf{PPP}}
\newcommand{\Tho}{\mathsf{T}}
\newcommand{\Pa}{\mathsf{p}}
\newcommand{\Event}{\mathsf{E}}
\newcommand{\Sij}{\mathsf{S}}
\newcommand{\s}{\mathsf{s}}
\newcommand{\Dp}{\mathsf{P}}
\newcommand{\Da}{\mathsf{d}}
\newcommand{\ie}{\textit{i.e.}, }
\newcommand{\x}{\mathbf{x}}
\newcommand{\y}{\mathbf{y}}
\newcommand{\z}{\mathbf{z}}
\newcommand{\origin}{\mathbf{0}}
\newcommand{\ob}{\rm o}
\newcommand{\bp}{\mathbf{b}}
\newcommand{\bps}{\mathbf{c}}
\newcommand{\mbar}{\bar{m}}
\newcommand{\tbar}{\bar{t}}
\newcommand{\ball}[1]{\mathcal{B}\left(#1\right)}
\newcommand{\Prob}{\mathbb{P}}
\newcommand{\Expect}{\mathbb{E}}
\newcommand{\Probb}{\mathbb{P}}
\newcommand{\Tdp}{\mathsf{P}}
\newcommand{\rthree}{\mathbb{R}^3}
\newcommand{\diff}{\mathrm{d}}
\newcommand{\erfc}{\mathrm{erfc}}
\newtheorem{corollary}{Corollary}
\newtheorem{lemma}{Lemma}

%% file: ClusteredTD.bib
@article{saha2017enriched,
  title={Enriched $ K $-tier HetNet model to enable the analysis of user-centric small cell deployments},
  author={Saha, Chiranjib and Afshang, Mehrnaz and Dhillon, Harpreet S},
  journal={IEEE Trans. Wireless Commun.},
  volume={16},
  number={3},
  pages={1593--1608},
  year={2017},
  publisher={IEEE}
}

@article{Afshang2017Matern,
  title = {Nearest-Neighbor and Contact Distance Distributions for {M}at{\'e}rn Cluster Process},
  author = {Afshang, Mehrnaz and Saha, Chiranjib and Dhillon, Harpreet S.},
  year = {2017},
  month = dec,
  journal = {IEEE Commun. Lett.},
  volume = {21},
  number = {12},
  pages = {2686--2689},
  issn = {1089-7798},
  doi = {10.1109/LCOMM.2017.2747510},
  urldate = {2025-02-28},
  copyright = {https://ieeexplore.ieee.org/Xplorehelp/downloads/license-information/IEEE.html}
}

@article{Afshang2017Thomas,
  title={Nearest-neighbor and contact distance distributions for {T}homas cluster process},
  author={Afshang, Mehrnaz and Saha, Chiranjib and Dhillon, Harpreet S},
  journal={IEEE Wirel. Commun. Lett.},
  volume={6},
  number={1},
  pages={130--133},
  year={2016},
  publisher={IEEE}
}

@article{albeverio1998,
  title = {Brownian Survival in a Clusterized Trapping Medium},
  author = {Albeverio, Sergio and Bogachev, Leonid V.},
  year = {1998},
  month = feb,
  journal = {Rev. Math. Phys.},
  volume = {10},
  number = {02},
  pages = {147--189},
  issn = {0129-055X, 1793-6659},
  doi = {10.1142/S0129055X98000057},
  urldate = {2025-02-21},
  langid = {english}
}

@book{andrews2023a,
  title={An introduction to cellular network analysis using stochastic geometry},
  author={Andrews, Jeffrey G and Gupta, Abhishek K and Alammouri, Ahmad and Dhillon, Harpreet S},
  year={2023},
  publisher={Springer}
}

@article{azimi-abarghouyi2023,
  title = {Fundamentals of Clustered Molecular Nanonetworks},
  author = {{Azimi-Abarghouyi}, Seyed Mohammad and Dhillon, Harpreet S. and Tassiulas, Leandros},
  year = {2023},
  month = jun,
  journal = {IEEE Trans. Mol. Biol. Multi-Scale Commun.},
  volume = {9},
  number = {2},
  pages = {135--145},
  issn = {2372-2061, 2332-7804},
  doi = {10.1109/TMBMC.2023.3267399},
  urldate = {2024-12-11},
  copyright = {https://ieeexplore.ieee.org/Xplorehelp/downloads/license-information/IEEE.html}
}

@article{Azimi2018SingleCluster,
  title={Stochastic geometry modeling and analysis of single-and multi-cluster wireless networks},
  author={Azimi-Abarghouyi, Seyed Mohammad and Makki, Behrooz and Haenggi, Martin and Nasiri-Kenari, Masoumeh and Svensson, Tommy},
  journal={IEEE Trans. Commun.},
  volume={66},
  number={10},
  pages={4981--4996},
  year={2018},
  publisher={IEEE}
}

@article{azimi2024matern,
  title = {Mat{\'e}rn Cluster Process with Holes at the Cluster Centers},
  author = {{Azimi-Abarghouyi}, Seyed Mohammad and Dhillon, Harpreet S.},
  year = {2024},
  month = jan,
  journal = {Statistics \& Probability Letters},
  volume = {204},
  pages = {109931},
  issn = {01677152},
  doi = {10.1016/j.spl.2023.109931},
  urldate = {2025-02-28},
  langid = {english}
}

@article{berezhkovskii1989,
  title = {Wiener Sausage Volume Moments},
  author = {Berezhkovskii, A. M. and Makhnovskii, {\relax Yu}. A. and Suris, R. A.},
  year = {1989},
  month = oct,
  journal = {J Stat Phys},
  volume = {57},
  number = {1-2},
  pages = {333--346},
  issn = {0022-4715, 1572-9613},
  doi = {10.1007/BF01023647},
  urldate = {2025-02-21},
  copyright = {http://www.springer.com/tdm},
  langid = {english}
}

@book{chiu2013stochastic,
  title={Stochastic geometry and its applications},
  author={Chiu, Sung Nok and Stoyan, Dietrich and Kendall, Wilfrid S and Mecke, Joseph},
  year={2013},
  publisher={John Wiley \& Sons}
}

@article{chude-okonkwo2017,
  title = {Molecular Communication and Nanonetwork for Targeted Drug Delivery: A Survey},
  shorttitle = {Molecular Communication and Nanonetwork for Targeted Drug Delivery},
  author = {{Chude-Okonkwo}, Uche A. K. and Malekian, Reza and Maharaj, B. T. and Vasilakos, Athanasios V.},
  year = 2017,
  journal = {IEEE Commun. Surv. Tutorials},
  volume = {19},
  number = {4},
  pages = {3046--3096},
  issn = {1553-877X},
  doi = {10.1109/COMST.2017.2705740},
  urldate = {2025-02-25},
  copyright = {https://ieeexplore.ieee.org/Xplorehelp/downloads/license-information/IEEE.html}
}

@article{deng2017a,
  title = {Analyzing Large-Scale Multiuser Molecular Communication via 3-{{D}} Stochastic Geometry},
  author = {Deng, Yansha and Noel, Adam and Guo, Weisi and Nallanathan, Arumugam and Elkashlan, Maged},
  year = {2017},
  month = jun,
  journal = {IEEE Trans. Mol. Biol. Multi-Scale Commun.},
  volume = {3},
  number = {2},
  pages = {118--133},
  issn = {2372-2061, 2332-7804},
  doi = {10.1109/TMBMC.2017.2750145},
  urldate = {2023-09-12}
}

@article{dissanayake2019a,
  title = {Interference Mitigation in Large-Scale Multiuser Molecular Communication},
  author = {Dissanayake, Maheshi Buddhinee and Deng, Yansha and Nallanathan, Arumugam and Elkashlan, Maged and Mitra, Urbashi},
  year = {2019},
  month = jun,
  journal = {IEEE Trans. Commun.},
  volume = {67},
  number = {6},
  pages = {4088--4103},
  issn = {0090-6778, 1558-0857},
  doi = {10.1109/TCOMM.2019.2897568},
  urldate = {2023-09-13}
}

@article{farsad2016,
  title = {A Comprehensive Survey of Recent Advancements in Molecular Communication},
  author = {Farsad, Nariman and Yilmaz, H. Birkan and Eckford, Andrew and Chae, Chan-Byoung and Guo, Weisi},
  year = 2016,
  journal = {IEEE Commun. Surv. Tutorials},
  volume = {18},
  number = {3},
  pages = {1887--1919},
  issn = {1553-877X},
  doi = {10.1109/COMST.2016.2527741},
  urldate = {2025-02-25},
  copyright = {https://ieeexplore.ieee.org/Xplorehelp/downloads/license-information/OAPA.html}
}

@book{haenggi2012,
  title = {Stochastic Geometry for Wireless Networks},
  author = {Haenggi, Martin},
  year = {2012},
  month = oct,
  edition = {1},
  publisher = {Cambridge University Press},
  doi = {10.1017/CBO9781139043816},
  urldate = {2023-09-13},
  isbn = {978-1-107-01469-5 978-1-139-04381-6}
}

@book{kern1934,
  title = {Solid Mensuration},
  author = {Kern, W.F. and Bland, J.R.},
  year = {1934},
  publisher = {J. Wiley \& sons},
  address = {United Kingdom}
}

@inproceedings{mosayebi2018,
  title = {Advanced Target Detection via Molecular Communication},
  booktitle = {Proc. {{GLOBECOM}}},
  author = {Mosayebi, Reza and Wicke, Wayan and Jamali, Vahid and Ahmadzadeh, Arman and Schober, Robert and {Nasiri-Kenari}, Masoumeh},
  year = {2018},
  month = dec,
  pages = {1--7},
  publisher = {IEEE},
  address = {Abu Dhabi, United Arab Emirates},
  doi = {10.1109/GLOCOM.2018.8647734},
  urldate = {2023-09-12},
  isbn = {978-1-5386-4727-1}
}

@article{nakano2012b,
  title = {Molecular Communication and Networking: {{Opportunities}} and Challenges},
  shorttitle = {Molecular Communication and Networking},
  author = {Nakano, T. and Moore, M. J. and {Fang Wei} and Vasilakos, A. V. and {Jianwei Shuai}},
  year = {2012},
  month = jun,
  journal = {IEEE Trans. Nanobioscience},
  volume = {11},
  number = {2},
  pages = {135--148},
  issn = {1536-1241, 1558-2639},
  doi = {10.1109/TNB.2012.2191570},
  urldate = {2025-02-25},
  copyright = {https://ieeexplore.ieee.org/Xplorehelp/downloads/license-information/IEEE.html}
}

@book{nakano2013c,
  title = {Molecular Communication},
  author = {Nakano, Tadashi and Eckford, Andrew W. and Haraguchi, Tokuko},
  year = {2013},
  month = sep,
  edition = {1},
  publisher = {Cambridge University Press},
  doi = {10.1017/CBO9781139149693},
  urldate = {2025-02-25},
  copyright = {https://www.cambridge.org/core/terms},
  isbn = {978-1-107-02308-6 978-1-139-14969-3}
}

@article{nakano2017a,
  title = {Performance Evaluation of Leader--Follower-Based Mobile Molecular Communication Networks for Target Detection Applications},
  author = {Nakano, Tadashi and Okaie, Yutaka and Kobayashi, Shouhei and Koujin, Takako and Chan, Chen-Hao and Hsu, Yu-Hsiang and Obuchi, Takuya and Hara, Takahiro and Hiraoka, Yasushi and Haraguchi, Tokuko},
  year = {2017},
  month = feb,
  journal = {IEEE Trans. Commun.},
  volume = {65},
  number = {2},
  pages = {663--676},
  issn = {0090-6778},
  doi = {10.1109/TCOMM.2016.2628037},
  urldate = {2023-09-12}
}

@book{okaie2016a,
  title={Target detection and tracking by bionanosensor networks},
  author={Okaie, Yutaka and Nakano, Tadashi and Hara, Takahiro and Nishio, Shojiro},
  year={2016},
  publisher={Springer}
}

@article{Pandey2021Matern,
  title={kth distance distributions of n-dimensional {M}at{\'e}rn cluster process},
  author={Pandey, Kaushlendra and Gupta, Abhishek K},
  journal={IEEE Commun. Lett.},
  volume={25},
  number={3},
  pages={769--773},
  year={2020},
  publisher={IEEE}
}

@article{sabu2019,
  title = {Analysis of Diffusion Based Molecular Communication with Multiple Transmitters Having Individual Random Information Bits},
  author = {Sabu, Nithin V. and Gupta, Abhishek K.},
  year = {2019},
  month = dec,
  journal = {IEEE Trans. Mol. Biol. Multi-Scale Commun.},
  volume = {5},
  number = {3},
  pages = {176--188},
  issn = {2372-2061, 2332-7804},
  doi = {10.1109/TMBMC.2020.2986719},
  urldate = {2023-09-12}
}

@article{sabu2020b,
  title = {Detection Probability in a Molecular Communication via Diffusion System with Multiple Fully-Absorbing Receivers},
  author = {Sabu, Nithin V. and Gupta, Abhishek K.},
  year = {2020},
  month = dec,
  journal = {IEEE Commun. Lett.},
  volume = {24},
  number = {12},
  pages = {2824--2828},
  issn = {1089-7798, 1558-2558, 2373-7891},
  doi = {10.1109/LCOMM.2020.3017243},
  urldate = {2023-09-12}
}

@article{sabu2024d,
  title = {On the Target Detection Performance of a Molecular Communication Network with Multiple Mobile Nanomachines},
  author = {Sabu, Nithin V. and Gupta, Abhishek K.},
  year = {2024},
  month = jul,
  journal = {IEEE Trans.on Nanobioscience},
  volume = {23},
  number = {3},
  pages = {524--536},
  issn = {1536-1241, 1558-2639},
  doi = {10.1109/TNB.2024.3399188},
  urldate = {2025-02-25},
  copyright = {https://ieeexplore.ieee.org/Xplorehelp/downloads/license-information/IEEE.html}
}

@article{yilmaz2014b,
  title = {Three-Dimensional Channel Characteristics for Molecular Communications with an Absorbing Receiver},
  author = {Yilmaz, H. Birkan and Heren, Akif Cem and Tugcu, Tuna and Chae, Chan-Byoung},
  year = {2014},
  month = jun,
  journal = {IEEE Commun. Lett.},
  volume = {18},
  number = {6},
  pages = {929--932},
  issn = {1089-7798},
  doi = {10.1109/LCOMM.2014.2320917},
  urldate = {2025-02-28},
  copyright = {https://ieeexplore.ieee.org/Xplorehelp/downloads/license-information/IEEE.html}
}
